\documentclass[11pt]{article}

\usepackage{amsmath,amsthm,amscd,amsfonts,amssymb}
\usepackage{mathrsfs}
\usepackage{hyperref}
\usepackage[all]{xy}
\usepackage[usenames,dvipsnames,svgnames,table]{xcolor}

\definecolor{darkred}{RGB}{175,0,0}

\newtheorem{theorem}{Theorem}
\newtheorem{proposition}{Proposition}

\newtheorem{lemma}{Lemma}
\newtheorem{definition}{Definition}

\newtheorem{problem}{Problem}

\theoremstyle{definition}
\newtheorem*{remarkth}{Remark}
\newenvironment{remark}{\begin{remarkth}}{\hfill$\lozenge$\end{remarkth}}

\newcommand{\vf}{\mathfrak{X}}
\newcommand{\df}{\Omega}
\newcommand{\Leg}{\mathcal{F}}
\newcommand{\leg}{\operatorname{leg}}
\newcommand{\F}{\mathbb{F}}
\renewcommand{\d}{d}

\newcommand{\R}{\mathbb{R}}
\newcommand{\Id}{\textnormal{Id}}
\renewcommand{\Im}{\operatorname{Im}}

\newcommand{\pr}{\operatorname{pr}}
\newcommand{\derpar}[2]{\displaystyle\frac{\partial{#1}}{\partial{#2}}}
\newcommand{\derpars}[3]{\displaystyle\frac{\partial^2{#1}}{\partial{#2}\partial{#3}}}
\newcommand{\restric}[2]{\left.#1\right|_{#2}}
\def\Lie{\mathop{\rm L}\nolimits}
\def\inn{\mathop{i}\nolimits}
\def\tabaddress#1{{\small\it\begin{tabular}[t]{c}#1
\\[1.2ex]\end{tabular}}}

\title{REGULARITY PROPERTIES OF FIBER DERIVATIVES ASSOCIATED WITH HIGHER-ORDER MECHANICAL SYSTEMS}

\author{
{\sc  Leonardo Colombo\thanks{\textbf{e}-{\it mail}: ljcolomb@umich.edu} }\\
\vspace{5mm}
   \tabaddress{Department of Mathematics, University of Michigan. \\
   Ann Arbor, MI 48109, USA} \\
{\sc  Pedro Daniel Prieto-Mart\'{\i}nez\thanks{\textbf{e}-{\it mail}: math.pedro.daniel.prieto@gmail.com} }\\
\vspace{5mm}
\tabaddress{Departament de Matem\`atiques. Building C3, North Campus UPC. \\
   Jordi Girona 1. 08034 Barcelona. Spain}\\
}
\date{\today}

\begin{document}

\maketitle

\pagestyle{myheadings}

\thispagestyle{empty}

\begin{abstract}
The aim of this work is to study fiber derivatives associated to Lagrangian and Hamiltonian
functions describing the dynamics of a higher-order autonomous dynamical system. More precisely,
given a function in $T^*T^{(k-1)}Q$, we find necessary and sufficient conditions for such a
function to describe the dynamics of a $k$th-order autonomous dynamical system, thus being a
$k$th-order Hamiltonian function. Then, we give a suitable definition of (hyper)regularity for
these higher-order Hamiltonian functions in terms of their fiber derivative. In addition, we
also study an alternative characterization of the dynamics in Lagrangian submanifolds in terms
of the solutions of the higher-order Euler-Lagrange equations.
\end{abstract}

\bigskip
\noindent \textbf{Key words}:
\textsl{Higher-order systems; Lagrangian and Hamiltonian mechanics; Lagrangian submanifolds; Tulczyjew's triple.}

\vbox{\raggedleft AMS s.\,c.\,(2010): 70H50, 53D05, 53D12}\null 
\markright{\rm L. Colombo, P.D. Prieto-Mart\'{\i}nez: \textsl{Regularity properties of higher-order systems}}

\clearpage

\tableofcontents

\section{Introduction}

Higher-order dynamical systems play a relevant role in certain branches of theoretical physics,
applied mathematics and numerical analysis. In particular, they appear in theoretical physics, in
the mathematical description of relativistic particles with spin, string theories, Hilbert's
Lagrangian for gravitation, Podolsky's generalization of electromagnetism and others, as well as
in some problems of fluid mechanics and classical physics, and in numerical models arising from the
discretization of dynamical control systems that preserve their inherent geometric structures.
In these kinds of systems, the dynamics have explicit dependence on accelerations or higher-order
derivatives of the generalized coordinates of position. The geometric tools used to study those
systems have been developed mainly by M. de Le\'on, P.R. Rodr\'igues, D.J. Saunders and M. Crampin
(among others) between  the $70$'s and $90$'s in
\cite{book:DeLeon_Rodrigues85,book:Saunders89,art:Saunders_Crampin90}
(see also \cite{art:Batlle_Gomis_Pons_Roman88,proc:Cantrijn_Crampin_Sarlet86,art:Carinena_Lopez92,
art:Colombo_Martin_Zuccalli10,art:Gracia_Pons_Roman91,art:Gracia_Pons_Roman92,art:Krupkova94,
phd:Martinez}, and references therein). These works are based in the ideas of the Lagrangian formalism
introduced by J. Klein at the beginning of the $60$'s in \cite{art:Klein62}. In the aforementioned work,
the Euler-Lagrange equations are obtained by J. Klein in a purely geometric way using the canonical
geometric structures of the tangent and cotangent bundles, avoiding the use of variational calculus and exploiting the geometry of these dynamical systems.

The interest in higher-order dynamical systems has been growing up from the 90's due to the study
of optimization and boundary value problems where the cost function involves higher-order derivatives,
which may be modeled as variational problems with explicit dependence on higher-order derivatives of
the generalized coordinates of position. These ``higher-order variational problems'' are of great
interest for their useful applications in aeronautics, robotics, computer-aided design, air traffic
control, trajectory planning, and, more generally, problems of interpolation and approximation of
curves on Riemannian manifolds. These kinds of problems have been studied in
\cite{art:Bloch_Colombo_Gupta_Martin15,art:BlochCrouch96,art:Camarinha_Silva_Crouch95,
art:Hussein_Bloch07,art:Machado_Silva_Krakowski10,art:Maruskin_Bloch11,art:Noakes_Heinzinger_Paden89} and more recently, in  \cite{art:Gay_Holm_Meier_Ratiu_Vialard12_1,art:Gay_Holm_Meier_Ratiu_Vialard12_2,art:Gay_Holm_Ratiu11,
phd:Meier} the development of variational principles which involve higher-order cost functions for optimization problems on Lie groups and their application in template matching for computational anatomy have been studied. These applications have produced a great interest in the study and development of new modern geometric tools and techniques to model properly higher-order variational problems, with the
additional goal of obtaining a deepest understanding of the intrinsic properties of these problems.
Some work in this line of research has been carried out recently in the following references, \cite{art:Colombo_Prieto14, art:Colombo_Martin14, art:Prieto_Roman11,art:Prieto_Roman13, art:Martinez15, Popescu1, Popescu, art:Bruce_Grabowska_Grabowski_Urbanski15, art:Jozwikowski_Rotkiewicz13,art:Jozwikowski_Rotkiewicz15, art:Grabowska_Vitagliano15,art:Vitagliano10}.

Let us recall that the dynamics for a $k$th-order dynamical system can be obtained both by means
of a Lagrangian function defined on a the $k$th-order tangent bundle $T^{(k)}Q$ of the smooth
manifold $Q$ that models the configuration space of the system, or by means of a Hamiltonian
function defined on the cotangent bundle $T^*T^{(k-1)}Q$ (see \cite{book:DeLeon_Rodrigues85} for
details). The relation between Lagrangian and Hamiltonian dynamics can be studied either
using the higher-order Legendre map \cite{art:DeLeon_Lacomba88,art:DeLeon_Lacomba89} or the
Legendre-Ostrogradsky map \cite{book:DeLeon_Rodrigues85}. These two well-known transformations are
both derived from the Lagrangian function, and they provide a way to define a canonical Hamiltonian
function, and hence to give a Hamiltonian formulation of the system. Nevertheless, for first-order
dynamical systems it is known from \cite{book:Abraham_Marsden78} (Sections \S3.5 and \S3.6) that,
starting from a Hamiltonian function in the cotangent bundle of the configuration space, it is
possible to define a Lagrangian function in the tangent bundle describing the dynamics of the system,
or to recover the starting Lagrangian if the Hamiltonian was defined using a Legendre map. The
fundamental tool to do so is the fiber derivative of the Hamiltonian function. This same procedure
can be carried out with a Hamiltonian function defined on $T^*T^{(k-1)}Q$, and a ``Lagrangian
function'' can be defined in $TT^{(k-1)}Q$. This ``Lagrangian function'', however, may not be an
actual Lagrangian function in the physical sense, since the generalized coordinates in the base
manifold $T^{(k-1)}Q$ are generally not related to the fibered coordinates (the ``velocities'').
From the geometric point of view, this pretended ``Lagrangian function'' may not be defined in the
holonomic submanifold $T^{(k)}Q \hookrightarrow TT^{(k-1)}Q$, which is the real domain for a
$k$th-order Lagrangian function. This problem arises from the fact that the starting Hamiltonian
function is not considered as a Hamiltonian function for $k$th-order system, but just as a Hamiltonian
function for a first-order system defined on the cotangent bundle of a larger base manifold, since
there is not an enforced relation between the momenta.

The discussion in the previous paragraph gives rises to two natural questions. First, what is a
\textit{higher-order Hamiltonian function}? And second, what does it mean for such a function to
be \textit{regular}? Observe that these concepts are clearly defined for first-order Lagrangian
or Hamiltonian functions, and also for a higher-order Lagrangian functions (see \cite{Fer} for the constrained first order case). In this work we pretend
to give an answer to these questions. Indeed, in Section \ref{sec:Definition&RegularityHOHamiltonian}
we propose a definition of both concepts, always taking into account the particular case of
Hamiltonian functions associated to a regular Lagrangian system. Moreover, we extend in a nontrivial
way some results from \cite{book:Abraham_Marsden78}, namely Propositions 3.6.7 and 3.6.8, and
Theorem 3.6.9, to higher-order autonomous dynamical systems by means of the answers to the proposed
questions.

For constrained systems, as well for singular Lagrangian functions, an alternative approach for a
better understanding of the geometry involving mechanical systems was established by W.M. Tulczyjew:
the so-called Tulczyjew's triple \cite{art:Tulczyjew76_2,art:Tulczyjew76_1}, which makes strong use
of Lagrangian submanifolds of suitable symplectic manifolds. Lagrangian submanifolds are of great
interest in geometric mechanics, since they provide a way of describing both Lagrangian and
Hamiltonian dynamics from a purely geometric and intrinsic point of view (see \cite{art:Tulczyjew76_2, art:Tulczyjew76_1}). In particular, let us recall that given a mechanical
system described by a Lagrangian function $L \colon TQ \to \R$, the Lagrangian dynamics are
``generated'' by the Lagrangian submanifold $dL(TQ) \subset T^*TQ$. On the other hand, if the
system is described by a Hamiltonian function $H \colon T^*Q \to \R$, then the Hamiltonian dynamics
are ``generated'' by the Lagrangian submanifold $dH(T^*Q) \subset T^*T^*Q$. The relationship
between these two formulations is provided by the so-called Tulczyjew's triple
\begin{equation*}
\xymatrix{
T^{*}TQ & \ & TT^{*}Q \ar[ll]_-{\alpha_{Q}}\ar[rr]^-{\beta_Q} & \ & T^{*}T^*Q,
}
\end{equation*}
where $\alpha_Q$ and $\beta_Q$ are both vector bundle isomorphisms, and $T^{*}TQ$, $TT^{*}Q$
and $T^{*}T^*Q$ are double vector bundles equipped with suitable symplectic structures.

If we consider a $k$th-order Lagrangian system described by a $k$th-order Lagrangian function
$L \colon T^{(k)}Q \to \R$, a similar construction can be carried out (with some additional
technical issues arising from the fact that $T^{(k)}Q$ is not a vector bundle in general, see
\cite{art:DeLeon_Lacomba88,art:DeLeon_Lacomba89} for details), thus obtaining a Lagrangian
submanifold in $TT^*T^{(k-1)}Q$, which completely determines the equations of motion for the
dynamics. Moreover, these equations of motion are of Hamiltonian type if the Lagrangian system
is regular. Our aim in this work is to study properties of fiber derivatives of functions defined on Lagrangian submanifolds, thus pursuing the research lines established in the works of M. de Le\'on and E. Lacomba \cite{art:DeLeon_Lacomba88,art:DeLeon_Lacomba89}.

The paper is structured as follows. In Section \ref{sec:GeometricBackground} we introduce the
necessary geometric background in order to make the paper as much selfcontained as possible.
In particular, we review the definition and basic properties of symplectic manifolds and Lagrangian
submanifolds, fiber derivatives of fiber preserving maps and their application to relate Lagrangian
and Hamiltonian dynamics, higher-order tangent bundles and some of its canonical structures and,
finally, a short review on the construction of Tulczyjew's triple for first-order dynamical systems.
In Section \ref{sec:GeometricDescriptionHODS} we study fiber derivatives of higher-order
Lagrangian systems, relating the classical Legendre-Ostrogradsky map associated to a $k$th-order
Lagrangian function $L$ and the $k$th-order Legendre transformation defined on the Lagrangian
submanifold $\Sigma_{L}\subset T^{*}TT^{(k-1)}Q$ generated by the Lagrangian $L$. To close this
Section, we study the dynamics of the end of a thrown javelin as an illustrative example. Finally, in Section
\ref{sec:Definition&RegularityHOHamiltonian} we study the problem of giving an universal definition
of higher-order Hamiltonian function, and the regularity properties of the fiber derivative
associated with such a Hamiltonian.

\section{Geometric background}
\label{sec:GeometricBackground}

In this Section we introduce the geometric structures and definitions that we use along this work. All the manifolds are real, second countable and $C^\infty$. The maps and the structures are assumed
to be $C^\infty$. Sum over crossed repeated indices is understood. If $M$ denotes a finite-dimensional
smooth manifold, then $C^\infty(M)$, $\vf(M)$ and $\Omega^k(M)$ denote the sets of smooth functions,
smooth vector fields and smooth $k$-forms on $M$, respectively.

\subsection{Symplectic manifolds and Lagrangian submanifolds}

Along this Subsection, $M$ denotes a finite-dimensional smooth manifold.
We refer to \cite{book:Cannas01,book:Libermann_Marle87,book:Weinstein77} for details and proofs.


\begin{definition}
A \textnormal{symplectic form} in $M$ is a closed $2$-form $\omega \in \Omega^2(M)$ which is
nondegenerate, that is, for every $p \in M$, $\inn_{X_p}\omega_p = 0$ if, and only if, $X_p = 0$ where $X_p\in T_p M$.
A \textnormal{symplectic manifold} is a pair $(M,\omega)$, where $M$ is a smooth manifold and
$\omega$ is a symplectic form.
\end{definition}

\begin{remark}
If $(M,\omega)$ is a symplectic manifold, then the nondegeneracy of $\omega$ implies that $M$ has
even dimension, that is, $\dim M = 2n$.
\end{remark}

\begin{definition}
Let $(M_1,\omega_1)$ and $(M_2,\omega_2)$ be two symplectic manifolds, and $\Phi \colon M_1 \to M_2$
a diffeomorphism. $\Phi$ is a \textnormal{symplectomorphism} if $\Phi^*\omega_2 = \omega_1$, and it
is an \textnormal{anti-symplectomorphism} if $\Phi^*\omega_2 = -\omega_1$, where $\Phi^{*}\omega_2$
denotes the pull-back of the $2$-form $\omega_2$ by the diffeomorphism $\Phi$.
\end{definition}

A distinguished symplectic manifold is the cotangent bundle $T^*Q$ of a $n$-dimensional smooth
manifold $Q$. Let $\pi_Q \colon T^*Q \to Q$ be the canonical projection defined by
$\pi_Q(\alpha_q) = q\in Q$, where $\alpha_q\in T^*_{q}Q$. The \textsl{Liouville $1$-form}, denoted by $\theta_Q \in \df^1(T^*Q)$, is defined as
\begin{equation*}
\langle(\theta_Q)_{\alpha_q}, X_{\alpha_q} \rangle = \langle \alpha_q , T_{\alpha_q}\pi_{T^{*}Q}(X_{\alpha_q})\rangle \, , \
\mbox{where } \alpha_q\in T^*Q \mbox{ and } X_{\alpha_q} \in T_{\alpha_q}T^*Q \, .
\end{equation*}
Observe that the Liouville $1$-form satisfies $\alpha^*\theta_Q = \alpha$ for every
$\alpha \in \df^1(Q)$. Then, one can define the \textsl{canonical symplectic form} of $T^*_{q}Q$, or
\textsl{Liouville $2$-form}, as $\omega_Q = -\d\theta_Q \in \df^2(T^*Q)$. If $(U;(q^i))$,
$1 \leqslant i \leqslant n$, are local coordinates in $Q$, the induced natural coordinates in 
$\pi_Q^{-1}(U) \subseteq T^*Q$ are $(q^i,p_i)$, $1 \leqslant i \leqslant n$. In these coordinates,
the local expression of the Liouville $1$-form is $\theta_Q = p_i \d q^i$, from where, the canonical symplectic form has the following coordinate expression $\omega_Q = \d q^i \wedge \d p_i$.

The existence of a nondegenerate $2$-form on symplectic manifolds enables us to define some special 
submanifolds. In particular, we are interested in the study of \textsl{Lagrangian submanifolds}.

\begin{definition}
Let $(M,\omega)$ be a symplectic manifold. An immersed submanifold $i_N \colon N \hookrightarrow M$ is a
\textnormal{Lagrangian submanifold} if $\dim N = \frac{1}{2}( \dim M)$ and $i_N^*\omega = 0$.
\end{definition}

Next, we introduce some particular Lagrangian submanifolds of the symplectic manifold
$(T^*Q,\omega_Q)$. The first one is the image of a closed $1$-form. Indeed, let $\lambda\in\df^1(Q)$
be a closed $1$-form, and let us consider the subset
$\Sigma_{\lambda} = \lambda(Q) \subset T^*Q$, which is a submanifold of $T^*Q$ with canonical embedding $\lambda \colon Q \hookrightarrow T^*Q$.
Then we have $\lambda^*\omega_Q = \lambda^*(-\d\theta_Q) = -\d\lambda^*\theta_Q = -\d \lambda = 0$ since $\lambda$ is closed. Hence, $\Sigma_{\lambda}$ is a Lagrangian submanifold. If, moreover,
$\lambda$ is exact, that is, $\lambda = \d f$, with $f \in C^\infty(Q)$, we say that $f$ is a
\textsl{generating function} of the Lagrangian submanifold ${\Sigma}_{\lambda}$, and we denote it
by $\Sigma_{f}$ (see \cite{book:Weinstein77} for details).

There is a more general construction of Lagrangian submanifolds given by J. \'Sniatycki and W.M.
Tulczyjew in \cite{art:Sniatycki_Tulczyjew72} (see also \cite{art:Tulczyjew76_2,art:Tulczyjew76_1})
which we use in Subsection \ref{subsec:HOTulczyjewTriple} to generate the dynamics of a higher-order
dynamical system through Lagrangian submanifolds.

\begin{theorem}[\'Sniatycki \& Tulczyjew]\label{thm:Sniatycki_Tulczyjew}
Let $Q$ be a smooth manifold, $\tau_Q \colon TQ \to Q$ its tangent bundle,
$i_N \colon N \hookrightarrow Q$ a $k$-dimensional submanifold, and $f \colon N \to \R$ a smooth function.
Then
\begin{align*}
\Sigma _{f,N} &= \left\{ \mu \in T^*Q \mid \pi _{Q} (\mu) \in N \mbox{ and }
\left\langle \mu, v \right \rangle = \left\langle \d f , v \right\rangle \mbox{ for every } v \in T_{\pi_Q(\mu)}N \right\} \\
&= \left\{ \mu \in T^*Q \mid i_{N}^{*}\mu = \d f \right\} \, ,
\end{align*}
is a Lagrangian submanifold of $(T^*Q,\omega_Q)$.
\end{theorem}

Let $(q^i)$, $1 \leqslant i \leqslant n$, be local coordinates in $Q$ adapted to $N$,
that is, such that $N$ is locally defined by the constraints $q^{k+1} = \ldots = q^n = 0$. Then, the
smooth function $f \colon N \to \R$ depends only on the coordinates $q^{1},\ldots,q^k$, and the
submanifold $\Sigma_{f,N} \hookrightarrow T^*Q$ is locally defined by
\begin{equation*}
\Sigma_{f,N} = \left\{ (q^i,p_i) \in T^*Q \mid q^{i} = 0 \, , \
p_j = \derpar{f}{q^j} \mbox{ for } k+1 \leqslant i \leqslant n-k \, , \, 1 \leqslant j \leqslant k \right\}.
\end{equation*}
Thus, it follows that $\dim\Sigma_{f,N} = n = \dim Q = \frac{1}{2} \dim T^*Q.$ Moreover, taking into
account the local expression of the canonical symplectic form $\omega_{Q}$, if we denote
$i_{\Sigma_{f,N}} \colon \Sigma_{f,N} \hookrightarrow T^{*}Q$ the canonical embedding, it follows that
$i_{\Sigma_{f,N}}^{*}\omega_{Q} = 0$. Therefore, $\Sigma_{f,N}$ is a Lagrangian submanifold of the
symplectic manifold $(T^{*}Q,\omega_{Q})$ (see \cite{art:Tulczyjew76_2} for an intrinsic proof).

The importance of this result lies in the fact that Lagrangian submanifolds are associated to the
dynamics of Lagrangian and Hamiltonian systems subject or not to constraints as we show in
Subsection \ref{subsec:HOTulczyjewTriple}.

\subsection{Fiber derivative of a Hamiltonian function}
\label{subsec:FiberDerivativeHamiltonian}

Along this Subsection, we consider a first-order dynamical system with $n$ degrees of freedom whose
configuration space is modeled by a $n$-dimensional smooth manifold $Q$, and let $H \in C^\infty(T^*Q)$
be a Hamiltonian function describing the dynamics of the system (see \cite{book:Abraham_Marsden78},
\S 3.5 and \S 3.6, for details). 

\begin{definition}\label{def:FiberDerivative}
Let $\pi \colon E \to M$ and $\rho \colon F \to M$ be vector bundles over the common base manifold
$M$, and let $f \colon E \to F$ be a smooth fiber preserving map (not necessarily a vector bundle
morphism). Let $f_x$ denote $\restric{f}{E_x}$, where $E_x = \pi^{-1}(x)$ is the fiber over $x \in M$.
The \textnormal{fiber derivative} of $f$ is defined to be the map
\begin{equation*}
\begin{array}{rcl}
\Leg f \colon E & \longrightarrow & \displaystyle{\bigcup_{x \in M}} L(E_x,F_x) \\
v_x & \longmapsto & Df_x(v_x)
\end{array}
\end{equation*}
where $L(E_x,F_x)$ denotes the vector space of linear mappings from $E_x$ to $F_x$.
\end{definition}

Next, we apply the construction given in Definition \ref{def:FiberDerivative} to the Hamiltonian
function $H$. Let $E$ be the cotangent bundle of $Q$, $\pi_Q \colon T^*Q \to Q$ the canonical
projection, $F=Q\times\R$ the trivial vector bundle with projection
$\pr_1 \colon Q \times \R \to Q$, and $f$ the map $\widetilde{H} \colon T^*Q \to Q \times \R$
defined by
\begin{equation*}
\widetilde{H}(\alpha_q) = ( \pi_Q(\alpha_q), H(\alpha_q) ) \, .
\end{equation*}
The map $\widetilde{H}$ is smooth and fiber preserving, since
\begin{equation*}
\pr_1(\widetilde{H}(\alpha_q)) = \pr_1(\pi_Q(\alpha_q),H(\alpha_q)) = \pi_Q(\alpha_q) \, ,
\end{equation*}
although $\widetilde{H}$ is not a vector bundle morphism in general. Taking into account that
$L(T^*_qQ,\R) = T_q^{**}Q \cong T_qQ$, the \textsl{fiber derivative of $H$}, denoted by
$\Leg H \colon T^*Q \to TQ$, is defined as the fiber derivative of the map $\widetilde{H}$ in the
sense of Definition \ref{def:FiberDerivative}.

\begin{remark}
This same procedure can be carried out with a Lagrangian function $L \in C^\infty(TQ)$. As it
is well-known, the fiber derivative of $L$, $\Leg L \colon TQ \to T^*Q$, is the Legendre
map $\leg_L$ relating the Lagrangian and Hamiltonian formalisms of dynamical systems (see
\cite{book:Abraham_Marsden78}, \S3.5 for details).
\end{remark}

The map $\Leg H \colon T^*Q \to TQ$ is smooth and fiber preserving, that is, $\tau_Q \circ \Leg H = \pi_Q$.
Let $(U,(q^i))$ be a local chart in $Q$, and $(q^i,p_i)$ the
induced natural coordinates in $\pi_Q^{-1}(U) \subseteq T^*Q$. Then, the coordinate expression of
$\Leg H$ is determined by
\begin{equation*}
\Leg H(q^{i},p_{i})=\left(q^{i},\frac{\partial H}{\partial p_{i}}\right) \, ,
\end{equation*}
from where we can observe that $\Leg H$ is smooth and fiber preserving.

\begin{definition}\label{def:RegularHamiltonian}
A Hamiltonian function $H \in C^\infty(T^*Q)$ is \textnormal{regular} if the map $\Leg H \colon T^*Q
\to TQ$ is a local diffeomorphism, and it is \textnormal{hyperregular} if $\Leg H$ is a global
diffeomorphism. Otherwise, the Hamiltonian function is said to be \textnormal{singular}.
\end{definition}

Locally, the regularity condition for $H$ is equivalent to
\begin{equation*}
\det \left( \derpars{H}{p_i}{p_j} \right)(\alpha_q) \neq 0 \, , \
\mbox{for every } \alpha_q \in T^*_{q}Q \, ,
\end{equation*}
that is, a Hamiltonian function is regular if, and only if, its Hessian matrix with respect to the
momenta is invertible at every point of $T^*Q$.

Next, we give a brief review of the relation between Hamiltonian and Lagrangian formalisms in terms
of the fiber derivative of $H$ (see \cite{book:Abraham_Marsden78} details). First, let us recall
how to define a Lagrangian $L \in C^\infty(TQ)$ describing the dynamics of the system starting
from a Hamiltonian function.

\begin{proposition}[\cite{book:Abraham_Marsden78}, Prop. 3.6.7]\label{prop:LagrangianFromHamiltonian}
Let $H \in C^\infty(T^*Q)$ be a hyperregular Hamiltonian function, $\theta_Q \in \df^{1}(T^*Q)$
the Liouville $1$-form, $\omega_Q \in \df^{2}(T^*Q)$ the canonical symplectic form and $X_H \in
\vf(T^*Q)$ the unique vector field solution to the dynamical equation $\displaystyle{\inn_{X_H}\omega_Q = \d H \,}$. The function $L \in C^\infty(TQ)$ defined by $L = \theta_Q(X_H) \circ \Leg H^{-1} - H \circ \Leg H^{-1}$ is hyperregular, and $\Leg L \equiv \leg_L = \Leg H^{-1}$.
\end{proposition}	

Observe that, up to this point, the Hamiltonian function $H \in C^\infty(T^*Q)$ could be any
function defined on the cotangent bundle. Nevertheless, if $L \in C^\infty(TQ)$ is a
hyperregular Lagrangian function describing the dynamics of the system, and $E_L = \Delta(L)
- L \in C^\infty(TQ)$ is the energy of the system, with $\Delta \in \vf(TQ)$ being the Liouville
vector field, then we can define a Hamiltonian function $H = E_L \circ \leg_L^{-1} \in
C^\infty(T^*Q)$. Then, the following result holds.

\begin{proposition}[\cite{book:Abraham_Marsden78}. Prop. 3.6.8 and Thm. 3.6.9]
\label{prop:FiberDerivativeHInverse}
Let $L \in C^\infty(TQ)$ be a hyperregular Lagrangian and $H = E_L \circ \leg_L^{-1} \in
C^\infty(T^*Q)$ the associated Hamiltonian function. Then $H$ is hyperregular and
$\Leg H = \leg_L^{-1}$. In addition, if $\tilde{L} = \theta_Q(X_H) \circ \Leg H^{-1} -
H \circ \Leg H^{-1} \in C^\infty(TQ)$ is the hyperregular Lagrangian function associated to $H$
by Proposition \ref{prop:LagrangianFromHamiltonian}, then $\tilde{L} = L$.
\end{proposition}

\subsection{Tulczyjew's triple}
\label{subsec:TulczyjewTriple}

In \cite{art:Tulczyjew76_2,art:Tulczyjew76_1}, Tulczyjew established two identifications, the
first one between $TT^*Q$ and $T^*TQ$ (useful to describe Lagrangian mechanics) and the second one
between $TT^*Q$ and $T^*T^*Q$ (useful to describe Hamiltonian mechanics), giving rise to the so-called
\textsl{Tulczyjew's triple}. In this Subsection we summarize these results. Along this Subsection,
$Q$ denotes a $n$-dimensional smooth manifold.

Let us recall that the double tangent bundle $TTQ$ of a manifold $Q$ is endowed with two vector
bundle structures over the base $TQ$, given by the canonical projection $\tau_{TQ} \colon TTQ \to TQ$
arising from the tangent bundle structure, and the tangent map $T\tau_Q \colon TTQ \to TQ$ of $\tau_Q$
of the canonical projection $\tau_Q \colon TQ \to Q$ arising from the starting tangent bundle structure.
These two structures are related by the \textsl{canonical flip} $\kappa_Q \colon TTQ \to TTQ$, which is
an isomorphism of double vector bundles. If $(U;(q^i))$, $1 \leqslant i \leqslant n$, is a local
chart in $Q$ and $(q^i,v^i,\dot{q}^i,\dot{v}^i)$ the induced local coordinates in a suitable open
set of $TTQ$, then $\kappa_Q$ is given locally by
$
\kappa_Q(q^i,v^i,\dot{q}^i,\dot{v}^i) = (q^i,\dot{q}^i,v^i,\dot{v}^i)$.
It is clear from this coordinate expression that $\kappa_Q$ is an involution. From this, we can give
the following definition.

\begin{definition}
The \textnormal{Tulczyjew's isomorphisms} are the diffeomorphisms $\alpha_Q \colon TT^*Q \to T^*TQ$
and $\beta_Q \colon TT^*Q \to T^*T^*Q$ defined as follows:
\begin{enumerate}
\item $\alpha_Q$ is the dual map of $\kappa_Q$ (as a double vector bundle morphism).
\item If $\omega_Q \in \df^{2}(T^*Q)$ is the canonical symplectic form, then
$\beta_Q(X) = i_X\omega_Q \, , \ X\in TT^*Q $. \end{enumerate}
\end{definition}

Let $(U;(q^i))$ be a local chart in $Q$, and $(q^i,p_i)$ the induced natural coordinates in
$\pi_Q^{-1}(U) \subseteq T^*Q$. The induced natural coordinates in
$\tau_{T^*Q}^{-1}(\pi_Q^{-1}(U)) \subseteq TT^*Q$ are $(q^i,p_i,\dot{q}^i,\dot{p}_i)$. In these
coordinates, the maps $\alpha_Q$ and $\beta_Q$ are given by
$\alpha_{Q}(q^i,p_i,\dot{q}^i,\dot{p}_i) = (q^i,\dot{q}^i,\dot{p}_i,p_i)$ and 
$\beta_Q(q^i,p_i,\dot q^i,\dot p_i)=(q^i,\dot q^i,\dot p_i,p_i)$, respectively. 

The map $\alpha_Q$ is a symplectomorphism when we consider on $TT^{*}Q$ the symplectic structure given by the complete lift $\omega_{Q}^{c}$ of the canonical symplectic form $\omega_{Q}$ on $T^{*}Q$
and on $T^*TQ$ the canonical symplectic form $\omega_{TQ}$. On the other hand, the map $\beta_Q$ is
an anti-symplectomorphism when we consider on $TT^{*}Q$ the same symplectic structure $\omega_{Q}^{c}$
and we consider on $T^*T^*Q$ the canonical symplectic structure $\omega_{T^{*}Q}$.

The maps $\beta_Q$ and $\alpha_Q$ give rise to the \textsl{Tulczyjew triple}, summarized in the
following diagram
\begin{equation*}
\xymatrix{
T^*TQ \ar[ddr]_{\pi_{TQ}} & \ & TT^*Q \ar[ddr]^{\tau_{T^*Q}} \ar[ddl]_{T\pi_{Q}} \ar[ll]_{\alpha_Q}\ar[rr]^{\beta_Q} & \ & T^*T^*Q \ar[ddl]^{\pi_{T^*Q}}  \\
\ & \ & \ & \ & \ \\
\ & TQ \ar[rr]^{\leg_{{L}}} \ar[ddr]_{\tau_{Q}} & \ &T^*Q \ar[ddl]^{\pi_{Q}} & \ \\
\ & \ & \ & \ & \ \\
\ & \ & Q & \ & \
}
\end{equation*}
where $\leg_{L} \colon TQ \to T^*Q$ denotes the Legendre transformation associated to a given
Lagrangian function ${L} \in C^\infty(TQ)$.

\subsection{Higher-order tangent bundles}
\label{subsec:HOTangentBundles}

(See \cite{book:DeLeon_Rodrigues85,book:Saunders89} for details).

Let $Q$ be a $n$-dimensional smooth manifold. We introduce an equivalence relation in the set
$C^{\infty}(\R,Q)$ of smooth curves $\gamma \colon \R \to Q$ as follows: given two curves
$\gamma_1, \gamma_2 \colon (-a,a) \to Q$, with $a > 0$, we say that $\gamma_1$ and $\gamma_2$
have \textsl{contact of order  $k$} at $q_0 = \gamma_1(0) = \gamma_2(0)$ if there exists a local
chart $(U,\varphi)$ of $Q$ such that $q_0 \in U$ and
\begin{equation*}
\restric{\frac{d^j}{dt^j}}{t=0} \left(\varphi \circ \gamma_1(t)\right)
= \restric{\frac{d^j}{dt^j}}{t=0} \left(\varphi \circ\gamma_2(t)\right) \, ,
\end{equation*}
for $j = 0,\ldots,k$. This is a well defined equivalence relation in $C^{\infty}(\R,Q)$ and the
equivalence class of a  curve $\gamma$ is denoted $[\gamma ]_0^{(k)}$. The set of equivalence
classes is denoted $T^{(k)}Q$, and it can be proved that it is a smooth manifold. Moreover, the
map $\tau^k_{Q} \colon T^{(k)}Q \to Q$ defined by $\tau^k_{Q} \left([\gamma]_0^{(k)} \right) = \gamma(0)$
endows $T^{(k)}Q$ with a fiber bundle structure over $Q$, and therefore $T^{(k)}Q$ is called
the \textsl{tangent bundle of order $k$} of $Q$, or \textsl{$k$th-order tangent bundle} of $Q$.

The manifold $T^{(k)}Q$ is endowed with some additional structure. In particular, for every
$0 \leqslant r \leqslant k$ we define a surjective submersion
$\tau_Q^{(r,k)} \colon T^{(k)}Q \to T^{(r)}Q$ as
$\tau_Q^{(r,k)}\left([\gamma]_0^{(k)}\right) = [\gamma]_0^{(r)}$. It is easy to see that for every
$0 \leqslant r \leqslant k$, the map $\tau_Q^{(r,k)}$ defines a fiber bundle structure. Moreover,
we have that $T^{(1)} Q \equiv TQ$ is just the usual tangent bundle of $Q$, $T^{(0)} Q \equiv Q$
and $\tau_Q^{(0,k)} = \tau^k_{Q}$.

The \textsl{$r$-lift} of a smooth function $f \in C^\infty(Q)$, for $0 \leqslant r \leqslant k$, is
the smooth function $f^{(r,k)} \in C^\infty(T^{(k)}Q)$ defined as
\begin{equation*}
f^{(r,k)} \left( [\gamma]^{(k)}_0 \right) = \restric{\frac{d^r}{dt^r}}{t=0} \left(f \circ \gamma(t)\right) \, .
\end{equation*}
Of course, these definitions can be applied to functions defined on open sets of $Q$. Observe that
the $0$-lift of $f$ coincides with $f$.

Local coordinates in $T^{(k)}Q$ are introduced as follows. Let $(U,\varphi)$ a local chart in $Q$
with coordinates $(q^i)$, $1 \leqslant i \leqslant n$. Then the induced natural coordinates in the
open set $(\tau^k_Q)^{-1}(U) \equiv T^{(k)}U \subseteq T^{(k)}Q$ are
$\left( q_{(0)}^{i},q_{(1)}^{i},\ldots,q_{(k)}^{i} \right) \equiv \left( q^i_{(j)} \right)$ with
$1 \leqslant i \leqslant n$, $0 \leqslant j \leqslant k$, where $q_{(r)}^{i} = (q^i)^{(r,k)}$ for
$0 \leqslant r \leqslant k$. Sometimes, we use the standard conventions $q_{(0)}^{i}\equiv q^i$,
$q_{(1)}^{i}\equiv \dot{q}^i$ and $q_{(2)}^{i}\equiv \ddot{q}^i$.

The \textsl{canonical immersion} $j_k \colon T^{(k)}Q \to T(T^{(k-1)} Q)$ is defined as 
\begin{equation}\label{eqn:CanonicalImmersion}
j_k \left( [\gamma]_0^{(k)} \right) = [{\gamma}^{(k-1)}]_0^{(1)} \, ,
\end{equation}
where ${\gamma}^{(k-1)}$ is the lift of the curve $\gamma$ to $T^{(k-1)}Q$; that is, the curve
${\gamma}^{(k-1)} \colon \R \to T^{(k-1)}Q$ given by $\gamma^{(k-1)}(t) = [\gamma_t]_0^{(k-1)}$
where $\gamma_t(s) = \gamma(t+s)$. In the induced local coordinates of $T^{(k)}Q$, the map $j_k$ is
locally given by
\begin{equation*} 
j_k \left( q_{(0)}^{i},q_{(1)}^{i},q_{(2)}^{i},\ldots,q_{(k)}^{i} \right) =
\left( q_{(0)}^{i},q_{(1)}^{i},\ldots,q_{(k-1)}^{i};q_{(1)}^{i},q_{(2)}^{i},\ldots,q_{(k)}^{i} \right) \, ,
\end{equation*}
from where we can deduce that in the induced natural coordinates $\left( q^i_{(j)},v^i_{(j)} \right)$
of $TT^{(k-1)}Q$, the submanifold $T^{(k)}Q$ is defined locally by the $(k-1)n$ constraints
$v_{(j)}^i = q_{(j+1)}^i$.

Denote by $\df^{q}(T^{(k)}Q)$ the real vector space of $q$-forms on $T^{(k)}Q$. In the exterior
algebra of differential forms on $T^{(k)}Q$, denoted
$\displaystyle{\bigoplus_{q\geqslant 0}\df^{q}(T^{(k)}Q)}$, we define an equivalence relation as
follows: for $\alpha \in \df^{q}(T^{(k)}M)$ and $\beta \in \df^{q}(T^{(k')}Q)$,
\begin{equation*}
\alpha \sim \beta \Longleftrightarrow
\begin{cases} \alpha = \left( \tau_Q^{(k',k)} \right)^*(\beta) & \mbox{ if } k'\leqslant k \\ 
\beta = \left(\tau_Q^{(k,k')}\right)^*(\alpha) & \mbox{ if } k' \geqslant k.
\end{cases} \ 
\end{equation*}
Consider the quotient set $\displaystyle{{\Omega} = \bigoplus_{k \geqslant 0}\ df^{q}(T^{(k)}Q) \, / \sim}$,
which is a commutative graded algebra. In this set we define the \textsl{Tulczyjew's derivation},
denoted by $d_T$, as follows: for every $f \in C^\infty(T^{(k)}Q)$ the function
$d_Tf \in C^\infty(T^{(k+1)}Q)$ is defined as 
$\displaystyle{
d_Tf\left( [\gamma]^{(k+1)}_0 \right)
= \left\langle \d_{[\gamma]^{(k)}_0}f \, , j_{k+1} \left([\gamma]^{(k+1)}_0\right)\right\rangle}$ where $j_{k+1} \colon T^{(k+1)}Q \to T(T^{(k)}Q)$ is the canonical immersion, and the covector
$\d_{[\gamma]^{(k)}_0}f \in T^*_{[\gamma]^{(k)}_0}T^{(k)}Q$ is the exterior derivative of $f$ at
$[\gamma]^{(k)}_0 \in T^{(k)}Q$. Using the coordinate expression for $j_{k+1}$, the function $d_T$
is given locally by
\begin{equation*}
d_Tf\left(q_{(0)}^i,\ldots,q_{(k+1)}^i\right)
= \sum_{j=0}^{k} q_{(j+1)}^i \derpar{f}{q_{(j)}^i} \left( q_{(0)}^i,\ldots,q_{(k)}^i \right) \, .
\end{equation*}
The map $d_T$ extends to a derivation of degree $0$ in $\Omega$ and, as $d_T\d = \d d_T$, it
is determined by its action on functions and by the property $d_T(\d q_{(j)}^i) = \d q_{(j+1)}^i$.

\begin{definition}\label{defholonomic}
A curve $\psi \colon \R \to T^{(k)}Q$ is \textsl{holonomic of type $r$}, $1 \leqslant r \leqslant k$,
if $\phi^{(k-r+1,k)} = \tau_Q^{(k-r+1,k)} \circ \psi$, where $\phi = \tau_Q^{k} \circ \psi \colon \R \to Q$;
that is, the curve $\psi$ is the lifting of a curve in $Q$ up to $T^{(k-r+1)}Q$.
\end{definition}
From Definition \ref{defholonomic}, a vector field $X \in \vf(T^{(k)}Q)$ is a \textsl{semispray of type $r$},
$1 \leqslant r \leqslant k$, if every integral curve $\psi$ of $X$ is holonomic of type $r$.
In the natural cordinates of $T^{(k)}Q$, the local expression of a semispray of type $r$ is
\begin{equation}\label{eqn:SemisprayTypeRLocal}
X = q_{(1)}^i\derpar{}{q_{(0)}^i} + q_{(2)}^i\derpar{}{q_{(1)}^i} + \ldots + q_{(k-r+1)}^i\derpar{}{q_{(k-r)}^i} +
F_{(k-r+1)}^i\derpar{}{q_{(k-r+1)}^i} + \ldots + F_{(k)}^i\derpar{}{q_{(k)}^i} \, .
\end{equation}

\begin{remark}
It is clear that every holonomic curve of type $r$ is also holonomic of type $s$, for $s \geqslant r$.
The same remark is true for semisprays.
\end{remark}

\section{Geometric description of higher-order dynamical systems}
\label{sec:GeometricDescriptionHODS}

In this Section we aim at studying fiber derivatives of higher-order Lagrangian systems, relating the classical Legendre-Ostrogradsky map and the $k$th-order Legendre transformation.

\subsection{Higher-order Tulczyjew triple and dynamics generated by Lagrangian submanifolds}
\label{subsec:HOTulczyjewTriple}

In this Subsection we explore some new results in the construction of the Tulczyjew's triple for higher-order dynamical systems of M. de Le\'on and E. Lacomba   \cite{art:DeLeon_Lacomba88,art:DeLeon_Lacomba89}. In particular, we study fiber derivatives of higher-order
Lagrangian systems, relating the classical Legendre-Ostrogradsky map associated to a $k$th-order
Lagrangian function $L$ and the $k$th-order Legendre transformation defined on the Lagrangian
submanifold $\Sigma_{L}\subset T^{*}TT^{(k-1)}Q$ generated by the Lagrangian $L$. We show the theory with a simple but interesting example, the dynamics of the end of a thrown javelin.

\begin{definition}
The \textnormal{$k^{th}$-order Tulczyjew's isomorphism} is the map
$\beta_{T^{(k-1)}Q} \colon TT^{*}T^{(k-1)}Q \to T^{*}T^{*}T^{(k-1)}Q$ defined by $\beta_{T^{(k-1)}Q}(V):=i_{V}\omega_{T^{(k-1)}Q}$ with $V\in TT^{*}T^{(k-1)}Q$, and $\omega_{T^{(k-1)}Q}$ being the canonical symplectic form of $T^{*}T^{(k-1)}Q$.
\end{definition}

Let $(q^i)$, $1 \leqslant i \leqslant n$, be local coordinates in an open set $U \subset Q$, and
$(q_{(j)}^i)$, $0 \leqslant j \leqslant k-1$, the induced coordinates
in $(\pi_Q^{k-1})^{-1}(U) \subset T^{(k-1)}Q$ introduced in Section \ref{subsec:HOTangentBundles}. Then,
natural coordinates in $\left(\pi_Q^{k-1} \circ \pi_{T^{(k-1)}Q}\right)^{-1}(U) \subset T^*T^{(k-1)}Q$ are
$\left(q_{(j)}^i,p^{(j)}_i\right)$, from where we deduce that the induced local natural coordinates
in $TT^*T^{(k-1)}Q$ are $\left(q_{(j)}^i,p^{(j)}_i,\dot{q}_{(j)}^i,\dot{p}^{(j)}_i\right)$, with
$1 \leqslant i \leqslant n$ and $0 \leqslant j \leqslant k-1$. In these coordinates, the map
$\beta_{T^{(k-1)}Q}$ is locally given by
$\beta_{T^{(k-1)}Q}\left(q_{(j)}^i,p^{(j)}_i,\dot{q}_{(j)}^i,\dot{p}^{(j)}_i\right)
= \left( q_{(j)}^i,\dot{q}_{(j)}^i,\dot{p}^{(j)}_i,p^{(j)}_i \right)$.
This map is an anti-symplectomorphism when we consider $T^{*}T^{*}T^{(k-1)}Q$ endowed with the
canonical symplectic structure and $TT^*T^{(k-1)}Q$ endowed with the symplectic structure given by
the complete lift $\omega_{T^{(k-1)}Q}^{c}$ of the canonical symplectic form on $T^*T^{(k-1)}Q$.

The cotangent bundles $T^{*}T^{*}T^{(k-1)}Q$ and $T^{*}TT^{(k-1)}Q$ are examples of double vector
bundles (see \cite{art:Grabowski_Urbanski99} for details). In particular, the double vector bundles
$T^{*}T^{*}T^{(k-1)}Q$ and $T^{*}TT^{(k-1)}Q$ are canonically isomorphic via a vector bundle
isomorphism over $T^{*}T^{(k-1)}Q$
\begin{equation*} 
{\mathcal{R}}_{k} \colon T^{*}TT^{(k-1)}Q \to T^{*}T^{*}T^{(k-1)}Q \, .
\end{equation*}
This map is an anti-symplectomorphism of symplectic manifolds (considering in both cotangent bundles
the canonical symplectic structures), and also an isomorphism of double vector bundles. It is
completely determined by the condition
\begin{equation*}
\left\langle \mathcal{R}_k(\alpha_u),W_{T^{*}\tau_{T^{(k-1)}Q}(\alpha_u)} \right\rangle =
- \left\langle \alpha_u, \widetilde{W}_u \right\rangle
+ \left\langle W_{T^{*}\tau_{T^{(k-1)}Q}(\alpha_u)}, \widetilde{W}_{u} \right\rangle^{T} \, ,
\end{equation*}
for every $\alpha_u \in T^{*}_{u}TT^{(k-1)}Q$, $\widetilde{W}_{u} \in T_{u}TT^{(k-1)}Q$ and
$W_{T^{*}\tau_{T^{(k-1)}Q}(\alpha_u)}\in TT^{*}T^{(k-1)}Q$ satisfying the relation $\displaystyle{T\tau_{T^{(k-1)}Q}(\widetilde{W}_{u}) = T\pi_{T^{(k-1)}Q}(W_{T^{*}\tau_{Q}(\alpha_u)}) \,}$.

Here, $\langle \cdot,\cdot\rangle^{T} \colon TT^*T^{(k-1)}Q \times_{TT^{(k-1)}Q} TTT^{(k-1)}Q \to \R$
is the pairing defined by the tangent map of the usual pairing
$\langle \cdot,\cdot \rangle \colon T^*T^{(k-1)}Q \times_{T^{(k-1)}Q} TT^{(k-1)}Q \to \R$,
and the vector bundle projection $T^{*}\tau_{T^{(k-1)}Q} \colon T^{*}TT^{(k-1)}Q\to T^{*}T^{(k-1)}Q$
is characterized by
$\langle T^{*}\tau_{T^{(k-1)}Q}(\alpha_u),w\rangle=\langle\alpha_{u},w_{u}^{\vee}\rangle$ where $u, w \in T_{[q]_{(0)}^{(k-1)}}T^{(k-1)}Q$, $\alpha_u \in T^{*}_{u}TT^{(k-1)}Q$, and
$w_{u}^{\vee} \in T_{u}TT^{(k-1)}Q$ is the vertical lift of the tangent vector $w$ (see
\cite{art:Garcia_Guzman_Marrero_Mestdag14} for first order systems. The derivation for higher-order systems is derived straightforwardly from the definition given for first order systems).

Let $(q^i)$, $1 \leqslant i \leqslant n$, be local coordinates in $Q$, and $(q_{(j)}^i)$,
$1 \leqslant i \leqslant n$, $0 \leqslant j \leqslant k-1$, the induced coordinates in $T^{(k-1)}Q$
introduced in Section \ref{subsec:HOTangentBundles}. Then, natural coordinates in $TT^{(k-1)}Q$
are $\left(q_{(j)}^i,v_{(j)}^i\right)$, from where we deduce that the induced natural coordinates
in $T^*TT^{(k-1)}Q$ are $\left(q_{(j)}^i,v_{(j)}^i,p^{(j)}_i,\tilde{p}^{(j)}_i\right)$, with
$1 \leqslant i \leqslant n$ and $0 \leqslant j \leqslant k-1$. In these coordinates, the map
$\mathcal{R}_k$ is locally given by
\begin{equation*}
{\mathcal{R}}_k \left(q_{(j)}^i,v_{(j)}^i,p^{(j)}_i,\tilde{p}^{(j)}_i\right) =
\left(q_{(j)}^i,\tilde{p}^{(j)}_i,- p^{(j)}_i,v_{(j)}^i\right) \, .
\end{equation*}

Then, composing $\beta_{T^{(k-1)}Q}$ with $\mathcal{R}_k^{-1}$ we obtain a map
$\alpha_{T^{(k-1)}Q} \colon TT^*T^{(k-1)}Q \to T^*TT^{(k-1)}Q$, which is given in the natural coordinates
$\left(q_{(j)}^i,p^{(j)}_i,\dot{q}_{(j)}^i,\dot{p}^{(j)}_i\right)$ in $TT^*T^{(k-1)}Q$ by
\begin{equation}\label{eqn:HOAlphaLocal}
\alpha_{T^{(k-1)}Q}\left(q_{(j)}^i,p^{(j)}_i,\dot{q}_{(j)}^i,\dot{p}^{(j)}_i\right)
= \left( q_{(j)}^i,\dot{q}_{(j)}^i,\dot{p}^{(j)}_i,p^{(j)}_i \right) \, .
\end{equation}
This map is a symplectomorphism when we consider in $TT^*T^{(k-1)}Q$ the symplectic structure given
by the complete lift $\omega_{T^{(k-1)}Q}^c$ of the canonical symplectic form $\omega_{T^{(k-1)}Q}$
on $T^*T^{(k-1)}Q$ and on $T^*TT^{(k-1)}Q$ the canonical symplectic form $\omega_{TT^{(k-1)}Q}$.
The maps $\beta_{T^{(k-1)}Q}$ and $\alpha_{T^{(k-1)}Q}$ give rise to the \textsl{$k$th-order Tulczyjew triple}
\begin{equation*}
\xymatrix{
T^{*}TT^{(k-1)}Q & \ & TT^{*}T^{(k-1)}Q \ar[ll]_{\alpha_{T^{(k-1)}Q}}\ar[rr]^-{\beta_{T^{(k-1)}Q}} & \ & T^{*}T^*T^{(k-1)}Q,
}
\end{equation*}


\begin{remark}
The map $\alpha_{T^{(k-1)}Q} \colon TT^*T^{(k-1)}Q \to T^*TT^{(k-1)}Q$ can be obtained directly as the
dual to the \textsl{canonical flip} $\kappa_{T^{(k-1)}Q} \colon TT^{(k-1)}Q \to TT^{(k-1)}Q$, which
is an isomorphism of double vector bundle structures on $TT^{(k-1)}Q$ (see \cite{art:DeLeon_Lacomba89} for more details). We prefer to avoid the use of the canonical flip by using $\mathcal{R}_{k}^{-1}$, as in \cite{art:Grabowska_Vitagliano15}.
\end{remark}

Now we introduce the dynamics using a suitable Lagrangian submanifold in $T^*TT^{(k-1)}Q$ and the
$k$th-order Tulczyjew's triple. First, let $j_{k} \colon T^{(k)}Q \hookrightarrow TT^{(k-1)}Q$ be
the canonical immersion defined in \eqref{eqn:CanonicalImmersion}. Then, if $x \in T^{(k)}Q$, the
map $j_{k}^{*} \colon T_{j_{k}(x)}^{*}(TT^{(k-1)}Q) \to T_{x}^{*}(T^{(k)}Q)$ is given by
\begin{equation*}
j_{k}^{*}\mu = \mu\circ Tj_k \, , \ \mbox{for every } \mu \in T_{j_{k}(x)}^{*}TT^{(k-1)}Q \, .
\end{equation*}
Using this map, if $L\in C^\infty(T^{(k)}Q)$ is a $k$th-order Lagrangian function, by
\'Sniatycki and Tulczyjew's construction given in Theorem \ref{thm:Sniatycki_Tulczyjew} we define
a Lagrangian submanifold in the cotangent bundle $T^{*}TT^{(k-1)}Q$, endowed with the canonical
symplectic structure, as follows
\begin{equation*}\label{eqn:LagrangianSubmanifoldHO}
\Sigma_{L} = \left\{ \mu\in T^{*}TT^{(k-1)}Q \mid j_{k}^{*}\mu = \d L\right\}
\hookrightarrow T^{*}TT^{(k-1)}Q \, .
\end{equation*}

The Lagrangian submanifold $\Sigma_{L}$ fibers onto $j_{k}(T^{(k)}Q)$, it is locally parametrized by the $2kn$ coordinate functions
$\left( q_{(0)}^i,\ldots,q_{(k)}^i,\tilde{p}^{(0)}_i,\ldots,\tilde{p}^{(k-2)}_i \right)$,
$1 \leqslant i \leqslant n$, and it is immersed into $T^{*}TT^{(k-1)}Q$ as
\begin{equation*}
\left\{ \left( q_{(j)}^i \, ; \, q_{(j+1)}^i \, ; \,
\derpar{L}{q_{(0)}^i},\derpar{L}{q_{(1)}^i} - \tilde{p}^{(0)}_i,\ldots,\derpar{L}{q_{(k-1)}^i} - \tilde{p}^{(k-2)}_i \, ; \,
\tilde{p}^{(0)}_i,\ldots,\tilde{p}^{(k-2)}_i,\derpar{L}{q_{(k)}^i} \right)\right\} \, .
\end{equation*}

Therefore, taking this into account, the Lagrangian dynamics is given by the Lagrangian submanifold
$N_{L} = \alpha_{T^{(k-1)}Q}^{-1}(\Sigma_{L}) \hookrightarrow TT^{*}T^{(k-1)}Q$. Locally, $N_{L}$
is the set of elements in $TT^{*}T^{(k-1)}Q$ of the form
\begin{equation*}
\left(
q_{(j)}^i \, ; \,
\tilde{p}^{(0)}_i,\ldots,\tilde{p}^{(k-2)}_i,\derpar{L}{q_{(k)}^i} \, ; \,
q_{(j+1)}^i \, ; \,
\derpar{L}{q_{(0)}^i},\derpar{L}{q_{(1)}^i}-\tilde{p}^{(0)},\ldots,\derpar{L}{q_{(k-1)}^i}-\tilde{p}^{(k-2)}_i
\right) \, .
\end{equation*}
From this, the submanifold $N_L$ determines the following set of differential equations
\begin{align}
\frac{d}{dt}\tilde{p}^{(0)}_i = \derpar{L}{q_{(0)}^i} \, , \label{eqn:HOLagSub01}\\
\frac{d}{dt}\tilde{p}^{(j)}_i + \tilde{p}^{(j-1)}_i = \derpar{L}{q_{(j)}^i} \, ,\label{eqn:HOLagSub02}\\
\derpar{L}{q_{(k-1)}^i} - \tilde{p}^{(k-2)}_i = \frac{d}{dt}\left(\derpar{L}{q_{(k)}^i} \right) \, ,\label{eqn:HOLagSub03}
\end{align}
where $1 \leqslant j \leqslant k-2$ in \eqref{eqn:HOLagSub02}, and $1 \leqslant i \leqslant n$ in
every set. Differentiating the $n$ equations \eqref{eqn:HOLagSub03} with respect to the time $t$,
and replacing into equation \eqref{eqn:HOLagSub02} for $j = k - 2$, we obtain the
following equations
\begin{equation*}
\frac{d^{2}}{dt^{2}} \left( \derpar{L}{q_{(k)}^i} \right)
= \frac{d}{dt} \left( \derpar{L}{q_{(k-1)}^i} \right) - \derpar{L}{q_{(k-2)}^i} - \tilde{p}^{(k-3)}_i \, .
\end{equation*}
Differentiating the last set of equation with respect to the time $t$ and replacing the result into
\eqref{eqn:HOLagSub02} when $j = k - 3$ we have
\begin{equation*}
\frac{d^{3}}{dt^{3}} \left( \derpar{L}{q_{(k)}^i} \right)
= \frac{d^2}{dt^2} \left( \derpar{L}{q_{(k-1)}^i} \right)
- \frac{d}{dt} \left( \derpar{L}{q_{(k-2)}^i} \right)
+ \derpar{L}{q_{(k-3)}^i} - \tilde{p}^{(k-4)}_i \, .
\end{equation*}
Iterating the process $k-4$ times, we obtain the following set of $n$ equations
\begin{equation*}
\frac{d^{k}}{dt^{k}} \left( \derpar{L}{q_{(k)}^i}\right)
= \frac{d^{k-1}}{dt^{k-1}} \left( \derpar{L}{q_{(k-1)}^i} \right)
- \frac{d^{k-2}}{dt^{k-2}} \left( \derpar{L}{q_{(k-2)}^i} \right)
+ \ldots - \frac{d}{dt} \left( \derpar{L}{q_{(1)}^i} \right)
+ \frac{d}{dt} \tilde{p}^{(0)}_i \, .
\end{equation*}
Using equations \eqref{eqn:HOLagSub01} we obtain the following $n$ differential equations
\begin{equation*}
\sum_{j=0}^{k}(-1)^{j} \frac{d^{j}}{dt^{j}} \left( \derpar{L}{q_{(j)}^i} \right) = 0 \, ,
\end{equation*}
which are exactly the higher-order Euler-Lagrange equations for the higher-order Lagrangian function
$L$ (see \cite{art:DeLeon_Lacomba89}). From the computations and considerations given above, we have
the following result.

\begin{proposition}
The solutions of a $k$th-order Lagrangian system described by a $k$th-order Lagrangian function
$L \in C^\infty(T^{(k)}Q)$ are curves $\mu \colon I \subset \R \to \Sigma_{L}$
satisfying
\begin{equation*}
\restric{\pi_{TT^{(k-1)}Q}}{\Sigma_{L}} \circ \mu = j_k \circ \gamma^{(k)} \, ,
\end{equation*}
where $\gamma^{(k)} \colon I \to T^{(k)}Q$ is the $k$-lift of a curve $\gamma \colon I \to Q$,
$j_k \colon T^{(k)}Q \to TT^{(k-1)}Q$ the canonical immersion, and
$\restric{\pi_{TT^{(k-1)}Q}}{\Sigma_{L}} \colon \Sigma_{L} \to TT^{(k-1)}Q$ denotes the restriction
of the canonical projection $\pi_{TT^{k-1}Q} \colon T^*TT^{(k-1)}Q \to TT^{(k-1)}Q$ to $\Sigma_{L}$.
\end{proposition}

\begin{remark}
Observe that the three spaces $T(T^{*}T^{(k-1)}Q),$ $T^{*}(T^{*}T^{(k-1)}Q)$ and $T^{*}(T^{(k-1)}Q)$
involved in the Tulczyjew triple are symplectic manifolds; the two maps $\alpha_{T^{(k-1)}Q}$ and
$\beta_{T^{(k-1)}Q}$ involved in the construction are a symplectomorphism and an anti-symplectomorphism,
respectively; and the dynamical equations (Euler-Lagrange and Hamilton equations) are the local
equations defining the Lagrangian submanifolds
\begin{equation*}
N_L = \alpha_{T^{(k-1)}Q}^{-1}(\Sigma_{L}) \ \mbox{ and } \
S_H = \beta_{T^{(k-1)}Q}^{-1}\left(\d H(T^{*}(T^{(k-1)}Q))\right) \, ,
\end{equation*}
respectively. Moreover, the Lagrangian and Hamiltonian functions are not involved in the definition
of the triple. In this sense, the triple is canonical. Finally, we would like to point out that the
construction can be applied to an arbitrary Lagrangian function, not necessarily to a regular one.
\end{remark}

\begin{remark}[Higher-order variational constrained equations]
The natural extension for constrained (vakonomic) higher-order mechanical systems can be studied
by considering the embedded submanifold $\mathcal{M}\subset T^{(k)}Q$ given by the vanishing of $m$
independent constraint functions $\Phi^{\alpha} \colon T^{(k)}Q \to \R$, $\alpha=1,\ldots,m$. We consider the extended Lagrangian $\mathcal{L}=L+\lambda_{\alpha}\Phi^{\alpha}$ which includes the Lagrange multipliers $\lambda_{\alpha}$ as a new extra variable. The equations of motion for the higher-order constrined variational problem are the higher-order Euler-Lagrange equations for $\mathcal{L}$, that is, 
\begin{align}
\sum_{r=0}^{k}(-1)^r\frac{d^r}{dt^r}\left(\derpar{L}{q^{(r)i}}
+ \lambda_{\alpha}\derpar{\Phi^{\alpha}}{q^{(r)i}}\right)=&0 \\ 
\Phi^\alpha \left(q_{(0)}^i,\ldots,q^i_{(k-1)},q_{(k)}^i\right)=&0
\end{align}

From the geometrical point of view, these kind of higher-order variationally constrained problems are determined by a submanifold $\mathcal{M}\subset T^{(k)}Q$ with inclusion $i:\mathcal{M}\hookrightarrow T^{(k)}Q$ and by a Lagrangian $L_{\mathcal{M}}:\mathcal{M}\to\mathbb{R}$.  Using Theorem \eqref{thm:Sniatycki_Tulczyjew} we deduce that $\Sigma_{L_{\mathcal{M}}}\subset T^{*}TT^{(k-1)}Q$ is a Lagrangian submanifold. Moreover, using the Tulczyjew's symplectomorphism one induce a new Lagrangian submanifold $\alpha_{T^{(k-1)}Q}^{-1}(\Sigma_{L_{\mathcal{M}}})\subset TT^{*}T^{(k-1)}Q$ which completely determines the constrained variational dynamics. The case of unconstrained mechanics is generated taking the whole space $T^{(k)}Q$ instead of $\mathcal{M}$ and a Lagrangian function $L:T^{(k)}Q\to\mathbb{R}$. Indeed, this procedure gives the correct dynamics for the higher-order constrained variational problem using the same ideas as in the previous subsection only by changing the Lagrangian $L:T^{(k)}Q\to\R$ by $\mathcal{L}$. This is basically because the powerful of Tulczyjew's triple does not depends on the Lagrangian function.  


We assume that the restriction  of the projection $(\tau_Q^{(k-1, k)})\Big{|}_{{\mathcal M}}: {\mathcal M}\to T^{(k-1)}Q$ is a submersion. Locally, this conditions means that  the $m\times n$-matrix
$$\left(\frac{\partial(\Phi^{1},...,\Phi^{m})}{\partial(q^{1}_{(k)},...,q^{n}_{(k)})}\right)$$
 is of rank $m$ at all points of ${\mathcal M}$.

Consequently, by the implicit function theorem, we can locally express the constraints (reordering coordinates if necessary) as
\begin{equation}\label{constraints}
\phi^\alpha \left(q_{(0)}^i,\ldots,q^i_{(k-1)},q_{(k)}^a\right) = q_{(k)}^{\alpha} \, ,
\ 1 \leqslant \alpha \leqslant m \, , \ m+1 \leqslant a \leqslant n \, , \ i = 1,\ldots,n \, ,
\end{equation} and therefore, we can define a Lagrangian function $L \in C^\infty(\mathcal{M})$ in these new adapted coordinates to $\mathcal{M}$. Observe that
$\mathcal{M} \hookrightarrow T^{(k)}Q \hookrightarrow TT^{(k-1)}Q$, and let us denote by
$i_{\mathcal{M}} \colon \mathcal{M}\hookrightarrow TT^{(k-1)}Q$ the composition of both inclusions.
Now construct
\begin{equation*}
\Sigma_{L,\mathcal{M}} = \left\{\mu\in T^{*}TT^{(k-1)}Q\mid i_{\mathcal{M}}^{*}\mu = \d L\right\}
\end{equation*}
and proceeding as in the unconstrained case we obtain that the higher-order constrained dynamics
in governed by the following set of $n$ ordinary differential equations of order $2k$
\begin{equation}
\sum_{r=0}^{k}(-1)^r\frac{d^r}{dt^r}\left(\derpar{L}{q^{(r)i}}
- \mu_{\alpha}\derpar{\Phi^{\alpha}}{q^{(r)i}}\right) = 0 \, ,
\end{equation}
satisfying the constraints. Observe that $\mu\in T^{*}TT^{(k-1)}Q$ plays the
role of Lagrange multipliers for the equations, forcing the dynamics to satisfy the constraints
imposed by the submanifold $\mathcal{M}\subset T^{(k)}Q$.
\end{remark}
\subsection{On the Legendre maps for higher-order dynamical systems}


\label{subsec:HOLegendreMaps}

In this Subsection we introduce a Legendre transformation (a fiber derivative)
$\F L \colon \Sigma_L\to T^{*}T^{(k-1)}Q$ in the Lagrangian submanifold generated by a $k$th-order
Lagrangian function $L \in C^\infty(T^{(k)}Q)$ and we study its relationship with the
Legendre-Ostrogradsky map $\leg_{L} \colon T^{(2k-1)}Q \to T^*T^{(k-1)}Q$.

Let $Q$ be the configuration space of an autonomous dynamical system of order $k$ with $n$ degrees of freedom, and let $L\in C^\infty(T^{(k)}Q)$
be the Lagrangian function for this system. From the Lagrangian function $L$ we construct the
Poincar\'{e}-Cartan $1$-form $\theta_{L} \in \df^{1}(T^{(2k-1)}Q)$, whose coordinate expression is
\begin{equation}\label{eqn:PoincareCartan1FormLocal}
\theta_{L} = \sum_{r=1}^k \sum_{j=0}^{k-r}(-1)^j d_T^j\left(\derpar{L}{q_{(r+j)}^i}\right) \d q_{(r-1)}^i \, .
\end{equation}
The Poincar\'{e}-Cartan $1$-form $\theta_{L} \in \df^{1}(T^{(2k-1)}Q)$ allows to define the Legendre-Ostrogradsky map as follows:
\begin{definition}
The \textnormal{Legendre-Ostrogradsky map} associated to the $k$th-order Lagrangian function $L$
is the fiber bundle morphism $\leg_{L} \colon T^{(2k-1)}Q \to T^*T^{(k-1)}Q$ over $T^{(k-1)}Q$
defined as follows: for every $u \in TT^{(2k-1)}Q$,
\begin{equation}\label{eqn:LegendreOstrogradskyMapDef}
\theta_L(u) = \left\langle T \tau_Q^{(k-1,2k-1)}(u) \, , \, \leg_L(\tau_{T^{(2k-1)}Q}(u) \right\rangle \, .
\end{equation}
\end{definition}

Besides the condition $\pi_{T^{(k-1)}Q} \circ \leg_L = \tau_Q^{(k-1,2k-1)}$ stated in the
definition, the Legendre-Ostrogradsky map relates the Liouville form in $T^*T^{(k-1)}Q$ to the
Poincar\'{e}-Cartan $1$-form. That is, if $\theta_{T^{(k-1)}Q} \in \df^{1}(T^*T^{(k-1)}Q)$ is
the canonical form of the cotangent bundle $T^*T^{(k-1)}Q$, then
$\leg_L^*\theta_{T^{(k-1)}Q} = \theta_L$.

In the local coordinates $\left(q_{(j)}^i\right)$, $1 \leqslant i \leqslant n$,
$0 \leqslant j \leqslant 2k-1$, of $T^{(2k-1)}Q$ introduced in Section \ref{subsec:HOTangentBundles},
we define the following local functions
\begin{equation}\label{eqn:JacobiOstrogradskyMomenta}
\hat p^{(r-1)}_i = \sum_{j=0}^{k-r}(-1)^j d_T^j\left(\derpar{L}{q_{(r+j)}^i}\right) \, ,
\end{equation}
which are called the Jacobi-Ostrogradsky momenta. Observe that we have the following relation between
$\hat{p}^{(r)}_i$ and $\hat{p}^{(r-1)}_i$
\begin{equation}\label{eqn:HamHOMomentumCoordRelation}
\hat p^{(r-1)}_i = \derpar{L}{q_{(r)}^i} - d_T\left( \hat p^{(r)}_i \right) \, ,
\ \mbox{for } 1 \leqslant r \leqslant k-1 \, .
\end{equation}

\begin{remark}
The relation \eqref{eqn:HamHOMomentumCoordRelation} means that we can recover all the
Jacobi-Ostrogradsky momenta coordinates from the set of highest order momenta $(\hat p^{(k-1)}_i)$.
\end{remark}

Bearing in mind the local expression of the form $\theta_L$, we can write
$\theta_L = \hat p^{(j)}_i \d q_{(j)}^i$. Let $(U;(q^i))$, $1 \leqslant i \leqslant n$, be a local
chart of $Q$, and $\left(q_{(j)}^i\right)$, $0 \leqslant j \leqslant 2k-1$, the induced local
coordinates in $(\tau_Q^{2k-1})^{-1}(U) \subset T^{(2k-1)}Q$ introduced in Section
\ref{subsec:HOTangentBundles}. From this, it is clear that the local expression of the
Legendre-Ostrogradsky map $\leg_{L}$ is
\begin{equation*}\label{eqn:LegendreOstrogradskyMapLocal}
\leg_L^*\left(q_{(r-1)}^i\right) = q_{(r-1)}^i \quad ; \quad
\leg_L^*\left(p_i^{(r-1)}\right) = \hat{p}_i^{(r-1)} = \sum_{j=0}^{k-r}(-1)^j d_T^j\left(\derpar{L}{q_{(r+j)}^i}\right) \, ,
\end{equation*}
where $1 \leqslant r \leqslant k$, that is,
\begin{equation*}
\leg_L \left(q_{(0)}^i,\ldots,q_{(2k-1)}^i \right) =
\left( q_{(0)}^i,\ldots,q_{(k-1)}^i,\hat{p}_i^{(0)},\ldots,\hat{p}_i^{(k-1)}\right) \, .
\end{equation*}

Consider the tangent map $T\leg_L \colon T(T^{(2k-1)}Q) \to T(T^*T^{(k-1)}Q)$. Given an arbitrary
point $[\gamma]_0^{(2k-1)} \in T^{(2k-1)}Q$, the tangent map of $\leg_L$ at $[\gamma]_0^{(2k-1)}$ is
locally given by the following $2kn \times 2kn$ matrix
\begin{equation*}
T_{[\gamma]_0^{(2k-1)}}\leg_L =
\left(
\begin{array}{cccc|cccc}
\Id_n & \mathbf{0}_n & \ldots & \mathbf{0}_n & \mathbf{0}_n & \mathbf{0}_n & \ldots & \mathbf{0}_n \\
\mathbf{0}_n & \Id_n & \ldots & \mathbf{0}_n & \mathbf{0}_n & \mathbf{0}_n & \ldots & \mathbf{0}_n \\
\vdots & \vdots & \ddots & \vdots & \vdots & \vdots & \ddots & \vdots \\
\mathbf{0}_n & \mathbf{0}_n & \ldots & \Id_n & \mathbf{0}_n & \mathbf{0}_n & \ldots & \mathbf{0}_n \\ \hline
\ & \ & \ & \ & \ & \ & \ & \ \\[-10pt]
\derpar{\hat{p}_i^{(0)}}{q_{(0)}^j} & \derpar{\hat{p}_i^{(0)}}{q_{(1)}^j} & \ldots & \derpar{\hat{p}_i^{(0)}}{q_{(k-1)}^j} & \derpar{\hat{p}_i^{(0)}}{q_{(k)}^j} & \derpar{\hat{p}_i^{(0)}}{q_{(k+1)}^j} & \ldots & \derpar{\hat{p}_i^{(0)}}{q_{(2k-1)}^j} \\[15pt]
\derpar{\hat{p}_i^{(1)}}{q_{(0)}^j} & \derpar{\hat{p}_i^{(1)}}{q_{(1)}^j} & \ldots & \derpar{\hat{p}_i^{(1)}}{q_{(k-1)}^j} & \derpar{\hat{p}_i^{(1)}}{q_{(k)}^j} & \derpar{\hat{p}_i^{(1)}}{q_{(k+1)}^j} & \ldots & \derpar{\hat{p}_i^{(1)}}{q_{(2k-1)}^j} \\
\vdots & \vdots & \ddots & \vdots & \vdots & \vdots & \ddots & \vdots \\
\derpar{\hat{p}_i^{(k-1)}}{q_{(0)}^j} & \derpar{\hat{p}_i^{(k-1)}}{q_{(1)}^j} & \ldots & \derpar{\hat{p}_i^{(k-1)}}{q_{(k-1)}^j} & \derpar{\hat{p}_i^{(k-1)}}{q_{(k)}^j} & \derpar{\hat{p}_i^{(k-1)}}{q_{(k+1)}^j} & \ldots & \derpar{\hat{p}_i^{(k-1)}}{q_{(2k-1)}^j}
\end{array}
\right) \, ,
\end{equation*}
where $\Id_n$ denotes the $n \times n$ identity matrix, $\mathbf{0}_n$ the $n \times n$ null matrix
and $(\partial{\hat{p}_i^{(r)}}/\partial{q_{(s)}^j})$ the $n \times n$ Jacobian matrix of the vector
function $(\hat{p}_1^{(r)},\ldots,\hat{p}_n^{(r)})$ with respect to the $n$ variables
$(q_{(s)}^1,\ldots,q_{(s)}^n)$. Moreover, using the relation \eqref{eqn:HamHOMomentumCoordRelation}
among the momenta function in combination with the coordinate expression of the Tulczyjew's derivation,
a long but straightforward computation shows that
\begin{equation*}
\derpar{\hat{p}^{(s)}_i}{q_{(2k-1-s)}^j} = \derpars{L}{q_{(k)}^i}{q_{(k)}^j} \, ,
\end{equation*}
from where we deduce that the antidiagonal $n \times n$ blocks in the lower right submatrix of
$T_{[\gamma]_0^{(2k-1)}}\leg_L$ coincide with the Hessian matrix of the Lagrangian function.
Therefore, it is clear that the Lagrangian function $L \in C^\infty(T^{(k)}Q)$ is regular if, and
only if, the bundle morphism $\leg_L \colon T^{(2k-1)}Q \to T^*T^{(k-1)}Q$ is a local diffeomorphism.
As a consequence of this, we have that if $L$ is a $k$th-order regular Lagrangian then the set
$(q^A_{(i)},\hat p_A^{(i)})$, $0\leqslant i\leqslant k-1$, is a set of local coordinates in $T^{(2k-1)}Q$,
and $(\hat p^{(i)}_A)$ are called the \textsl{Jacobi-Ostrogradsky momenta coordinates}.

Now, we introduce a Legendre transformation $\F L \colon \Sigma_{L}\to T^{*}(T^{(k-1)}Q)$ in
the Lagrangian submanifold generated by a higher-order Lagrangian function.

\begin{definition}
The higher-order Legendre transformation on the Lagrangian submanifold $\Sigma_{L}$,
$\F L \colon \Sigma_{L}\to T^{*}T^{(k-1)}Q$, is the map defined by
$\F L = \restric{\tau_{T^{*}(T^{(k-1)}Q)}\circ (\alpha_{T^{k-1}Q})^{-1}}{\Sigma_{L}}$.
\end{definition}

In the natural coordinates
$\left( q_{(0)}^i,\ldots,q_{(k)}^i,\tilde{p}^{(0)}_i,\ldots,\tilde{p}^{(k-2)}_i \right)$
of $\Sigma_L$ introduced in Subsection \ref{subsec:HOTulczyjewTriple}, the map $\F L$ is locally determined by
\begin{equation*}
\F L ( q_{(0)}^i,\ldots,q_{(k)}^i,\tilde{p}^{(0)}_i,\ldots,\tilde{p}^{(k-2)}_i )
= \left(q_{(0)}^i,\ldots,q_{(k-1)}^i,\tilde{p}^{(0)}_i,\ldots,\tilde{p}^{(k-2)}_i,\derpar{L}{q_{(k)}^i}\right).
\end{equation*}
\begin{definition}
A higher-order Lagrangian system determined by $L:T^{(k)}Q\to\R$ is regular if, and only if, $\F L$ is a local diffeomorphism.
\end{definition}
\begin{remark}
Observe that a higher-order Lagrangian system is regular if and only if
$\left(\derpars{L}{q_{(k)}^i}{q_{(k)}^j}\right)$ is a nondegenerate matrix. In such a case, since
$\tilde{p}^{(k-1)}_i = \derpar{L}{q_{(k)}^i}$, by the implicit function theorem, we can define the $n$
coordinate functions $q_{(k)}^j$ as functions depending on
$q_{(0)}^i,\ldots,q_{(k-1)}^i,\tilde{p}^{(k-1)}_i$; that is,
\begin{equation}\label{eqn:ImplicitHighestVelocity}
\tilde{q}_{(k)}^{\,j} = f(q_{(0)}^i,\ldots,q_{(k-1)}^i,\tilde{p}^{(k-1)}_i) \, .
\end{equation}
\end{remark}

Using the higher-order Legendre transformation we can give in an alternative way the solutions of
the higher-order Lagrangian system as follows.

\begin{proposition}\label{prop:RelationLegendreMaps01}
The solutions of a $k$th-order Lagrangian system described by a $k$th-order Lagrangian function
$L \in C^\infty(T^{(k)}Q)$ are the curves $\mu \colon I\subset\R\to \Sigma_{L}$ satisfying
\begin{equation*}
\alpha_{T^{(k-1)}Q}^{-1}(\mu(t)) = \frac{d}{dt}\F L(\mu(t)),
\end{equation*}
where $\mu$ satisfies $\restric{\pi_{T^{*}TT^{k-1}Q}}{\Sigma_{L}}(\mu(t)) = \gamma^{(k)}(t)$, and where
$\gamma^{(k)}$ is the $k$-lift of a curve $\gamma \colon I \to Q$.
\end{proposition}

Next, we give an alternative characterization of the dynamics in the Lagrangian submanifold
$\Sigma_{L}$ in terms of the solution of the higher-order Euler-Lagrange equations.

\begin{proposition}\label{prop:RelationLegendreMaps02}
A curve $\gamma \colon I \to Q$ is a solution of the higher-order Euler-Lagrange equations derived
from a $k$th-order Lagrangian function $L \in C^\infty(T^{(k)}Q)$ if, and only if,
\begin{equation*}
\alpha_{T^{(k-1)}Q}\left(\frac{d}{dt} \left( \leg_{L} \circ \gamma^{(2k-1)} \right) \right) \in \Sigma_{L} \, ,
\end{equation*}
where $\gamma^{(2k-1)} \colon I \to T^{(2k-1)}Q$ is the $(2k-1)$-lift of $\gamma$ and
$\leg_L \colon T^{(2k-1)}Q\to T^{*}(T^{(k-1)}Q)$ is the Legendre-Ostrogradsky map defined
in \eqref{eqn:LegendreOstrogradskyMapDef}.
\end{proposition}
\begin{proof}
This proof is easy in coordinates. Let $(U;(q^i))$ be a local chart in $Q$, and let us denote
$\gamma(t) = (q^i(t))$ in $U$. Then, the $(2k-1)$-lift of $\gamma$ is given by
$\gamma^{(2k-1)}(t) = ( q_{(0)}^i(t),\ldots,q_{(2k-1)}^i(t))$, and we have
\begin{equation*}
\left( \leg_L \circ \gamma^{(2k-1)} \right)(t) =
\left( q_{(0)}^i(t),\ldots,q_{(k-1)}^i(t),\hat{p}^{(0)}_i(t),\ldots,\hat{p}^{(k-1)}_i(t) \right) \, ,
\end{equation*}
where $\hat{p}_{(r)}^i$, $1 \leqslant i \leqslant n$ and $0 \leqslant r \leqslant k-1$, are the
Jacobi-Ostrogradsky momenta coordinates defined in \eqref{eqn:JacobiOstrogradskyMomenta}.
Then, bearing in mind the coordinate expression \eqref{eqn:HOAlphaLocal} of $\alpha_{T^{k-1}Q}$ we have that
\begin{equation*}
\alpha_{T^{(k-1)}Q}\left(\frac{d}{dt} \left( \leg_{L} \circ \gamma^{(2k-1)} \right) \right) =
\left( q_{(j)}^i(t) ; q_{(j+1)}^i(t) ; \frac{d}{dt} \hat{p}_i^{(j)}(t) ; \hat{p}_i^{(j)} (t) \right) \, ,
\end{equation*}
with $1 \leqslant i \leqslant n$ and $0 \leqslant j \leqslant k-1$. Requiring 
$\alpha_{T^{(k-1)}Q}\left(\frac{d}{dt} \left( \leg_{L} \circ \gamma^{(2k-1)} \right) \right) \in \Sigma_L$,
we obtain the following system of $(k+1)n$ differential equations on the component functions of
$\gamma^{(2k-1)}$
$$
\frac{d}{dt}\hat{p}^{(0)}_i = \derpar{L}{q_{(0)}^i}, \quad 
\frac{d}{dt}\hat{p}^{(j)}_i + \hat{p}^{(j-1)}_i = \derpar{L}{ q_{(j)}^i},\quad
\hat{p}^{(k-1)}_i = \derpar{L}{q_{(k)}^i}$$
with $1 \leqslant j \leqslant k-1$ in the second set of equations (which, observe, is exactly
the relation \eqref{eqn:HamHOMomentumCoordRelation} among the momenta). Combining these equations
following the same patterns as in the end of Subsection \ref{subsec:HOTulczyjewTriple}, we obtain
the higher-order Euler-Lagrange equations
\begin{equation*}
\sum_{j=0}^{k}(-1)^{j} \frac{d^{j}}{dt^{j}} \left(\derpar{L}{q_{(j)}^i}\right) = 0 \, . \qedhere
\end{equation*}
\end{proof}

From Propositions \ref{prop:RelationLegendreMaps01} and \ref{prop:RelationLegendreMaps02} it is
clear that we can characterize the solutions of the Euler-Lagrange equations in purely geometric
way by means of the either the higher-order Legendre map $\F L$ or the Legendre-Ostrogradsky map
$\leg_L$. Now, assume that the Lagrangian function $L \in C^\infty(T^{(k)}Q)$ is regular, so both
maps $\F L \colon \Sigma_L\to T^*T^{(k-1)}Q$ and $\leg_L \colon T^{(2k-1)}Q\to T^*T^{(k-1)}Q$ are
local diffeomorphisms. Then consider the map $\phi_L \colon T^{(2k-1)}Q \to \Sigma_L$ defined by
the composition $\phi_L = \F L^{-1}\circ \leg_L$, which is locally given by
\begin{equation*}
\phi_{L}\left(q_{(0)}^{i},\ldots,q_{(2k-1)}^{i}\right)=\left(q_{(0)}^{i},\ldots,q_{(k-1)}^{i};
\hat{p}^{(0)}_i,\ldots,\hat{p}^{(k-2)}_i, \derpar{L}{q_{(k)}^i}\right) \, .
\end{equation*}
This map is a local diffeomorphism with local inverse $\phi_L^{-1} = \leg_L^{-1} \circ \F L$ in the
corresponding open sets. If, moreover, $L$ is hyperegular, then $\phi_L$ is a global diffeomorphism,
from which we can recover the dynamics using the implicit function theorem and the inverse of the
Tulczyjew's isomorphism. In addition, in this case we can establish the Hamiltonian formalism in
$T^*T^{(k-1)}Q$ by defining a \textsl{canonical Hamiltonian function} $H \in C^\infty(T^*T^{(k-1)}Q)$
with coordinate expression
\begin{equation*}
H(q_{(0)}^i,\ldots,q_{(k-1)}^i,p^{(0)}_i,\ldots,p^{(k-1)}_i)
= \sum_{j=0}^{k-2} p^{(j)}_iq_{(j+1)}^i + p^{(k-1)}_i\tilde{q}^{\,i}_{(k)} - L(q_{(0)}^i,\ldots,q_{(k-1)}^i,\tilde{q}^{\,i}_{(k)}) \, ,
\end{equation*}
where $\tilde{q}^{\,i}_{(k)}$ is given implicitly by \eqref{eqn:ImplicitHighestVelocity}. The
canonical Hamiltonian function obtained in this way does not depend on the choice of the Legendre
transformation used to derive it, and it depends only on the starting Lagrangian function. The
corresponding Hamiltonian vector field $X_{H}$ is determined by $\inn_{X_H}\omega_{T^{(k-1)}Q} = dH$.
In this case we have that
\begin{equation*}
\hbox{Im}(X_H)=X_{H}(T^{*}T^{(k-1)}Q)=\beta_{T^{k-1}Q}^{-1}(dH(T^{*}T^{(k-1)}Q))=\alpha_{T^{(k-1)}Q}^{-1}(\Sigma_L) \, .
\end{equation*}

In the singular case, the submanifold $\hbox{Im}(dH)$ is not transversal with respect to
$\pi_{T^*T^{(k-1)}Q}$, and therefore it is necessary to apply an integrability algorithm to find, if it
exists, a subset where there are consistent solutions of the dynamics (see \cite{phd:Gotay} and
\cite{art:Gotay_Nester_Hinds78}, for example).

\subsection{An example: the dynamics of the end of a javelin}
\label{sec:Example1}

Let us consider the dynamical system that describes the motion of the end of a thrown javelin.
This gives rise to a $3$-dimensional second-order dynamical system, which is a particular case of
the problem of determining the trajectory of a particle rotating about a translating center
\cite{art:Constantelos84}.
Let $Q = \R^3$ be the manifold modeling the configuration space for this system with coordinates
$(q_{(0)}^1,q_{(0)}^2,q_{(0)}^3) = (q_{(0)}^i)$. Using the induced coordinates in $T^{(2)}\R^3$,
the Lagrangian function for this system is
\begin{equation*}
L(q_{(0)}^i,q_{(1)}^i,q_{(2)}^i) = \frac{1}{2} \sum_{i=1}^{3} \left( \left(q_{(1)}^i\right)^2 - \left(q_{(2)}^i\right)^2\right) \, ,
\end{equation*}
which is a regular Lagrangian function since the Hessian matrix of $L$ with respect to the second-order
velocities is
\begin{equation*}
\left( \derpars{L}{q_{(2)}^j}{q_{(2)}^i} \right) = \begin{pmatrix} -1 & 0 & 0 \\ 0 & -1 & 0 \\ 0 & 0 & -1 \end{pmatrix} \, .
\end{equation*}

The second-order Euler-Lagrange equations are
\begin{equation*}
\frac{d^4}{dt^4}q_{(0)}^{i} + \frac{d^{2}}{dt^2}q_{(0)}^{i}=0 \, ,
\end{equation*}
for $1 \leqslant i \leqslant 3$. The general solution for the second-order Euler-Lagrange equations is
a curve $\gamma \colon I \subset \R \to \R^3$ with component functions
$\gamma(t) = (\gamma^1(t),\gamma^2(t),\gamma^3(t))$ given by
\begin{equation*}
\gamma^{i}(t) = c_1^i + c_2^i t + c_3^i\sin(t) + c_4^i \cos(t) \, ,
\end{equation*}
with $c_{j}^i$ constants, $1 \leqslant i \leqslant 3$ and $1 \leqslant j \leqslant 4$.

The Jacobi-Ostrogradsky momenta are given by
\begin{equation*}
\hat{p}^{(0)}_{i} = \derpar{L}{q_{(1)}^i} - \frac{d}{dt} \left( \derpar{L}{q_{(2)}^i} \right)
= q_{(1)}^i + q_{(3)}^i \, ,
\quad \hat{p}^{(1)}_{i} = \derpar{L}{q_{(2)}^i} = -q_{(2)}^i \, .
\end{equation*}
Therefore, the Legendre-Ostrogradsky map is given by
\begin{equation*}
\leg_L \left( q_{(0)}^i,q_{(1)}^i,q_{(2)}^i,q_{(3)}^i \right) =
\left(q_{(0)}^i,q_{(1)}^i,\hat{p}^{(0)}_{i},\hat{p}^{(1)}_i \right)
= \left(q_{(0)}^i,q_{(1)}^i,q_{(1)}^i + q_{(3)}^i,-q_{(2)}^i \right) \, .
\end{equation*}

Consider a Lagrangian submanifold $\Sigma_L \hookrightarrow T^*TTQ$ associated with the Lagrangian
function $L$. Local coordinates in the Lagrangian submanifold are
$(q_{(0)}^i,q_{(1)}^i,q_{(2)}^i,p^{(0)}_i)$. The Legendre transformation on $\Sigma_L$ is given locally
by
\begin{equation*}
\F L \left( q_{(0)}^i,q_{(1)}^i,q_{(2)}^i,p^{(0)}_i \right) =
\left(q_{(0)}^i,q_{(1)}^i,p^{(0)}_i, \derpar{L}{q_{(2)}^i}\right)
= \left(q_{(0)}^i,q_{(1)}^i,p^{(0)}_i,-q_{(2)}^i \right) \, .
\end{equation*}
$\F L$ is a diffeomorphism, and thus the second-order system is regular. Therefore, we can define
the second-order derivative in terms of the momenta as $q_{(2)}^{i} = -p^{(1)}_i$ and define a
Hamiltonian function on $T^{*}TQ$, which is given in coordinates by
\begin{equation*}
H \left( q_{(0)}^i,q_{(1)}^i;p^{(0)}_i,p^{(1)}_i \right) =
p^{(0)}_iq_{(1)}^i - \frac{1}{2}\left( \left(q_{(1)}^i\right)^2 + \left(p^{(1)}_i\right)^2 \right) \, .
\end{equation*}
Hamilton's equations for this second-order dynamical system are
\begin{equation*}
\frac{d}{dt}p^{(0)}_i = 0 \, ,
\quad \frac{d}{dt}p^{(1)}_i = -p^{(0)}_i + q_{(1)}^i \, ,
\quad \frac{d}{dt}q_{(0)}^i = q_{(1)}^i \, ,
\quad \frac{d}{dt}q_{(1)}^i = -p^{(1)}_i \, .
\end{equation*}
From the equations given above, in the regular case, one can obtain in a straightforward way the
second-order Euler-Lagrange equations.

\section{On the definition and the regularity of a higher-order Hamiltonian function}
\label{sec:Definition&RegularityHOHamiltonian}

In this Section we aim at studying the Hamiltonian functions that describe the dynamics of
higher-order systems. More particularly, we want to define the notions of ``higher-order
Hamiltonian function'' and ``regularity'' of these functions in a precise way.

Henceforth, we consider a $k$th-order dynamical system with $n$ degrees of freedom, and let $Q$ be
a $n$-dimensional smooth manifold modeling the configuration space of this system. From the results
in \cite{book:DeLeon_Rodrigues85} and in Section \ref{sec:GeometricDescriptionHODS} we know
that the Hamiltonian phase space for this system is the cotangent bundle $T^*T^{(k-1)}Q$. Hence, let
$H \in C^\infty(T^*T^{(k-1)}Q)$ be a Hamiltonian function describing the dynamics of the system.

\subsection{Statement of the problems}
\label{subsec:StatementProblems}

Following the patterns in Section \ref{subsec:FiberDerivativeHamiltonian}, let $\Leg H \colon
T^*T^{(k-1)}Q \to TT^{(k-1)}Q$ be the fiber derivative of $H$. If $(U;(q^i))$,
$1 \leqslant i \leqslant n$, is a local chart in $Q$ and $(q^i_{(j)},p_i^{(j)})$, $0 \leqslant j
\leqslant k-1$, are the induced local coordinates in a suitable open subset of $T^*T^{(k-1)}Q$,
then the map $\Leg H$ is given locally by
\begin{equation*}
\Leg H \left( q^i_{(0)},\ldots,q^i_{(k-1)};p_i^{(0)},\ldots,p_i^{(k-1)} \right)
= \left( q^i_{(0)},\ldots,q^i_{(k-1)};\derpar{H}{p_i^{(0)}},\ldots,\derpar{H}{p_i^{(k-1)}} \right) \, .
\end{equation*}
Moreover, the regularity condition for the Hamiltonian function $H$ is, locally, equivalent to
\begin{equation*}
\det \left( \derpars{H}{p^{(j)}_i}{p^{(j')}_{i'}} \right)(\alpha_q) \neq 0 \, , \
\mbox{for every } \alpha_q \in T^*T^{(k-1)}Q \, .
\end{equation*}

\begin{proposition}\label{prop:LagrangianFromHamiltonian1}
Let $H \in C^\infty(T^*T^{(k-1)}Q)$ be a hyperregular Hamiltonian function, $\theta_{T^{(k-1)}Q} \in \df^{1}(T^*T^{(k-1)}Q)$ is the Liouville $1$-form and
$X_H \in \vf(T^*T^{(k-1)}Q)$ is the unique vector field solution to the dynamical equation
\begin{equation*}
\inn_{X_H}\omega_{T^{(k-1)}Q} = \d H \, ,
\end{equation*}
where $\omega_{T^{(k-1)}Q} = -\d\theta_{T^{(k-1)}Q} \in \df^{2}(T^*T^{(k-1)}Q)$ is the canonical
symplectic form. Then there exists a global diffeomorphism $L \in C^\infty(TT^{(k-1)}Q)$, locally determined by \begin{equation*}L\left( q^i_{(j)};v^i_{(j)} \right) = \tilde{p}_i^{(j)} \dot{q}^i_{(j)}
- H \left( q^i_{(j)};\tilde{p}_i^{(j)} \right) \, ,
\end{equation*}
where $\tilde{p}_i^{(j)} = (\Leg H^{-1})^*p_i^{(j)}\in TT^{(k-1)}Q$.
\end{proposition}	

\begin{proof}
Using that the Hamiltonian function $H$ is hyperregular and using Proposition
\ref{prop:LagrangianFromHamiltonian} with $T^{(k-1)}Q$ as the base manifold, we define a
hyperregular Lagrangian function (a global diffeomorphism) ${L} \in C^\infty(TT^{(k-1)}Q)$ as follows
\begin{equation*}
{L} = \theta_{T^{(k-1)}Q}(X_H) \circ \Leg H^{-1} - H \circ \Leg H^{-1} \, ,
\end{equation*}
where $\theta_{T^{(k-1)}Q} \in \df^{1}(T^*T^{(k-1)}Q)$ is the Liouville $1$-form and
$X_H \in \vf(T^*T^{(k-1)}Q)$ is the unique vector field solution to the dynamical equation
\begin{equation*}
\inn_{X_H}\omega_{T^{(k-1)}Q} = \d H \, ,
\end{equation*}
where $\omega_{T^{(k-1)}Q} = -\d\theta_{T^{(k-1)}Q} \in \df^{2}(T^*T^{(k-1)}Q)$ is the canonical
symplectic form. In the induced local coordinates of $T^*T^{(k-1)}Q$, the Liouville $1$-form
is locally given by $\theta_{T^{(k-1)}Q} = p_i^{(j)} \d q^i_{(j)}$. From where we deduce that the
vector field $X_H \in \vf(T^*T^{(k-1)}Q)$ solution to the previous equation is locally given by
\begin{equation}\label{eqn:VectorFieldSolutionHamEq}
X_H = \derpar{H}{p_i^{(j)}} \derpar{}{q^i_{(j)}} - \derpar{H}{q^i_{(j)}} \derpar{}{p_i^{(j)}} \, .
\end{equation}
Then, the function $\Leg H^*{L} = \theta_{T^{(k-1)}Q}(X_H) - H \in C^\infty(T^*T^{(k-1)}Q)$ is
given in coordinates by
\begin{equation*}
\Leg H^*{L}\left( q^i_{(j)};p_i^{(j)} \right) = p_i^{(j)} \derpar{H}{p_i^{(j)}}
- H \left( q^i_{(j)};p_i^{(j)} \right) \, .
\end{equation*}
From this, the coordinate expression of the Lagrangian function ${L}$ in the induced natural
coordinates $\left( q^i_{(j)},v^i_{(j)} \right)$, $1 \leqslant i \leqslant n$,
$0 \leqslant j \leqslant k-1$, of $TT^{(k-1)}Q$ is
\begin{equation*}
L\left( q^i_{(j)};v^i_{(j)} \right) = \tilde{p}_i^{(j)} \dot{q}^i_{(j)}
- H \left( q^i_{(j)};\tilde{p}_i^{(j)} \right) \, ,
\end{equation*}
where $\tilde{p}_i^{(j)} = (\Leg H^{-1})^*p_i^{(j)}$ are local functions in $TT^{(k-1)}Q$.
\end{proof}

It is important to point out that the Lagrangian function obtained may not be a $k$th-order
Lagrangian function, in the physical sense: the coordinate functions in the basis $T^{(k-1)}Q$ may
not be well-related to the coordinate functions in the fibers, in the sense that, in general, we
have $q^i_{(j+1)} \neq \dot{q}^i_{(j)}$, that is, they are independent variables. From the geometric
point of view, the problem is that ${L}$ is not defined in the ``holonomic'' submanifold
$j_k \colon T^{(k)}Q \hookrightarrow TT^{(k-1)}Q$. In fact, observe that, up to this point, the
order of the system has not been taken into account at any step, since this condition is usually
inherited from the Lagrangian formulation. That is, the function $H \in C^\infty(T^*T^{(k-1)}Q)$
is not considered as a Hamiltonian function for a $k$th-order dynamical system, but just as a
Hamiltonian function for a first-order system defined on the cotangent bundle of a larger manifold,
since we do not require any relation among the momenta. This issue gives rise to the first problem
that we want to solve for Hamiltonian functions defined on $T^*T^{(k-1)}Q$.

\begin{problem}\label{prob:HOHamiltonianFunctions}
Given a function $H \in C^\infty(T^*T^{(k-1)}Q)$, find conditions on $H$ such that it describes
the dynamics of a $k$th-order dynamical system, that is, find a suitable definition of
\textnormal{$k$th-order Hamiltonian functions}.
\end{problem}

\begin{remark}
Observe that the equivalent problem in the Lagrangian formalism (find conditions on a function
${L} \in C^\infty(TT^{(k-1)}Q)$ such that ${L}$ is a $k$th-order Lagrangian function) is solved
straightforwardly, since there exists a distinguished ``holonomic'' submanifold
$j_k \colon T^{(k)}Q \hookrightarrow TT^{(k-1)}Q$. Nevertheless, there is not such a submanifold in
$T^*T^{(k-1)}Q$, and hence the problem is not trivial.
\end{remark}

Notice that a sufficient condition to ensure that ${L}$ is indeed a $k$th-order Lagrangian function
in the usual sense, that is, ${L} \in C^\infty(T^{(k)}Q)$ (as a submanifold of $TT^{(k-1)}Q$), is
to require $\Im(\Leg H) \subseteq j_k(T^{(k)}Q)$. However, since $\dim T^{(k)}Q = (k+1)n < 2kn =
\dim T^*T^{(k-1)}Q$ for $k > 1$, this requirement on the fiber derivative of $H$ prevents the
Hamiltonian function to be regular in the sense of Definition \ref{def:RegularHamiltonian}.
Nevertheless, recall that the regularity condition for a $k$th-order Lagrangian function
${L} \in C^\infty(T^{(k)}Q)$ is locally equivalent to
\begin{equation*}
\det \left( \derpars{{L}}{q^i_{(k)}}{q^j_{(k)}} \right) \left( [\gamma]_0^{(k)} \right) \neq 0 \, , \
\mbox{for every } [\gamma]_0^{(k)} \in T^{(k)}Q \, ,
\end{equation*}
with $1 \leqslant i,j \leqslant n$, instead of
\begin{equation*}
\det \left( \derpars{{L}}{q^i_{(r)}}{q^j_{(s)}} \right) \left( [\gamma]_0^{(k)} \right) \neq 0 \, , \
\mbox{for every } [\gamma]_0^{(k)} \in T^{(k)}Q \, ,
\end{equation*}
with $1 \leqslant i,j \leqslant n$ and $1 \leqslant r,s \leqslant k$. That is, the Hessian of ${L}$
is taken only with respect to the highest-order ``velocities'' $q^i_{(k)}$, and not with respect to
all the ``velocities''. Therefore, we deduce that the regularity condition given in Definition
\ref{def:RegularHamiltonian} is not suitable for higher-order Hamiltonian functions, since too many
``orders'' of the momenta coordinates are taken into account. This gives rise to the second problem
that we want to solve.

\begin{problem}\label{prob:HORegularHamiltonianFunctions}
To find a suitable definition of \textnormal{regularity} for $k$th-order Hamiltonian functions
$H \in C^\infty(T^*T^{(k-1)}Q)$ in terms of the fiber derivative of $H$ such that Propositions
\ref{prop:LagrangianFromHamiltonian} and \ref{prop:FiberDerivativeHInverse} hold for $k$th-order
dynamical systems.
\end{problem}

\subsection{A particular case: the Hamiltonian function associated to a (hyper)regular Lagrangian system}
\label{subsec:ParticularCase}

In order to solve Problems \ref{prob:HOHamiltonianFunctions} and \ref{prob:HORegularHamiltonianFunctions}
stated in the previous Subsection, we first consider the particular case of a well-known Hamiltonian
that describes properly the dynamics of a higher-order system: the Hamiltonian function associated
to a higher-order Lagrangian system. 

\begin{proposition} Given a hyperregular $k$th-order Lagrangian function, ${L} \in C^\infty(T^{(k)}Q)$, 
there exists a unique $k$th-order Hamiltonian function associated to this Lagrangian given
locally by
\begin{equation*}
H \left( q^i_{(0)},\ldots,q^i_{(k-1)};p_i^{(0)},\ldots,p_i^{(k-1)} \right)
= \sum_{j=0}^{k-2} q^i_{(j+1)}p_i^{(j)} + \tilde{q}^{\,i}_{(k)}p_i^{(k-1)}
- {L} \left( q^i_{(0)},\ldots,q^i_{(k-1)},\tilde{q}^{\,i}_{(k)} \right) \, ,
\end{equation*}
where $\tilde{q}^{\,i}_{(k)} = (\leg_{L}^{-1})^*q^i_{(k)}$. Moreover, $\Im(\Leg H) \subset j_{k}(T^{(k)}Q)$. 
\end{proposition}

\begin{proof}

Let $\leg_{L} \colon T^{(2k-1)}Q \to T^*T^{(k-1)}Q$ be the Legendre-Ostrogradsky map defined in
\eqref{eqn:LegendreOstrogradskyMapDef}. From the Lagrangian function ${L}$ we construct the
Lagrangian energy $E_{L} \in C^\infty(T^{(2k-1)}Q)$, with coordinate expression
\begin{equation}\label{eqn:LagrangianEnergyLocal}
E_{L} \left( q^i_{(0)},\ldots,q^i_{(2k-1)} \right)
= \sum_{r=1}^{k} q_{(r)}^i \sum_{j=0}^{k-r} (-1)^j d_T^j\left( \derpar{{L}}{q_{(r+j)}^i} \right)
- {L} \left( q^i_{(0)},\ldots,q^i_{(k)} \right) \, .
\end{equation}
Then, since ${L}$ is hyperregular, the Legendre-Ostrogradsky map is a diffeomorphism, and thus
there exists a unique $k$th-order Hamiltonian function associated to this Lagrangian system defined
by $H = E_{L} \circ \leg_{L}^{-1} \in C^\infty(T^*T^{(k-1)}Q)$. This Hamiltonian function is given
locally by
\begin{equation*}
H \left( q^i_{(0)},\ldots,q^i_{(k-1)};p_i^{(0)},\ldots,p_i^{(k-1)} \right)
= \sum_{j=0}^{k-2} q^i_{(j+1)}p_i^{(j)} + \tilde{q}^{\,i}_{(k)}p_i^{(k-1)}
- {L} \left( q^i_{(0)},\ldots,q^i_{(k-1)},\tilde{q}^{\,i}_{(k)} \right) \, ,
\end{equation*}
where $\tilde{q}^{\,i}_{(k)} = (\leg_{L}^{-1})^*q^i_{(k)}$ are local functions in $T^*T^{(k-1)}Q$.
Observe that the derivative of $H$ with respect to a momenta coordinate $p_r^{(j)}$ gives
\begin{equation*}
\derpar{H}{p_r^{(j)}} =
\begin{cases}
q^r_{(j+1)} 
& \mbox{if } 0 \leqslant j \leqslant k-2 \, , \\[15pt]
\tilde{q}^{\,r}_{(k)} + p_i^{(k-1)}\derpar{\tilde{q}^{\,i}_{(k)}}{p_r^{(k-1)}} - \derpar{{L}}{q^i_{(k)}}
\derpar{\tilde{q}^{\,i}_{(k)}}{p_r^{(k-1)}} = \tilde{q}^{\,r}_{(k)} & \mbox{if } j = k-1 \, ,
\end{cases}
\end{equation*}
where, since $T^*T^{(k-1)}Q = \Im(\leg_{L})$, we have $p^{(k-1)}_i = \derpar{{L}}{q^i_{(k)}}$,
and therefore the last two terms in the above sums cancel each other. From this we deduce that the
coordinate expression of the fiber derivative of $H$ in this particular case is
\begin{equation}\label{eqn:FiberDerivativeHamFromLagLocal}
\Leg H \left( q^i_{(0)},\ldots,q^i_{(k-1)};p_i^{(0)},\ldots,p_i^{(k-1)} \right)
= \left( q^i_{(0)},\ldots,q^i_{(k-1)}; q^i_{(1)},\ldots,q^i_{(k-1)},\tilde{q}^{\,i}_{(k)} \right) \, .
\end{equation}
It is clear from this coordinate expression that
$\Im(\Leg H) \subseteq T^{(k)}Q \stackrel{j_k}{\hookrightarrow} TT^{(k-1)}Q$. 
\end{proof}

Let 
$\Leg H_o \colon T^*T^{(k-1)}Q \to T^{(k)}Q$ be the map defined by $\Leg H = \Leg H_o \circ j_k$,
that is, the unique map such that the following diagram commutes
\begin{equation*}
\xymatrix{
T^*T^{(k-1)}Q \ar[rrr]^{\Leg H} \ar[drrr]_{\Leg H_o} & \ & \ & TT^{(k-1)}Q \\
\ & \ & \ & T^{(k)}Q \ar@{^{(}->}[u]^{j_k}
}
\end{equation*}
with coordinate expression
\begin{equation*}
\Leg H_o \left( q^i_{(0)},\ldots,q^i_{(k-1)};p_i^{(0)},\ldots,p_i^{(k-1)} \right)
= \left( q^i_{(0)},\ldots,q^i_{(k-1)},\tilde{q}^{\,i}_{(k)} \right) \, .
\end{equation*}

Now, let us compute the coordinate expression of the tangent map of $\Leg H$ in an arbitrary point
$\alpha_q \in T^*T^{(k-1)}Q$. Bearing in mind the coordinate expression
\eqref{eqn:FiberDerivativeHamFromLagLocal} of $\Leg H$, the map $T_{\alpha_q}\Leg H$ is given in
coordinates by the following $2kn \times 2kn$ real matrix
\begin{equation*}
T_{\alpha_q}\Leg H =
\left(
\begin{array}{cccc|ccccc}
\Id_n & \mathbf{0}_n & \ldots & \mathbf{0}_n & \mathbf{0}_n & \mathbf{0}_n & \ldots & \mathbf{0}_n & \mathbf{0}_n \\
\mathbf{0}_n & \Id_n & \ldots & \mathbf{0}_n & \mathbf{0}_n & \mathbf{0}_n & \ldots & \mathbf{0}_n & \mathbf{0}_n \\
\vdots & \vdots & \ddots & \vdots & \vdots & \vdots & \ddots & \vdots & \vdots \\
\mathbf{0}_n & \mathbf{0}_n & \ldots & \Id_n & \mathbf{0}_n & \mathbf{0}_n & \ldots & \mathbf{0}_n & \mathbf{0}_n \\ \hline
\mathbf{0}_n & \Id_n & \ldots & \mathbf{0}_n & \mathbf{0}_n & \mathbf{0}_n & \ldots & \mathbf{0}_n & \mathbf{0}_n \\
\vdots & \vdots & \ddots & \vdots & \vdots & \vdots & \ddots & \vdots & \vdots \\
\mathbf{0}_n & \mathbf{0}_n & \ldots & \Id_n & \mathbf{0}_n & \mathbf{0}_n & \ldots & \mathbf{0}_n & \mathbf{0}_n
\\
\derpar{\tilde{q}^{\,i}_{(k)}}{q^j_{(0)}} & \derpar{\tilde{q}^{\,i}_{(k)}}{q^j_{(1)}} & \ldots &
\derpar{\tilde{q}^{\,i}_{(k)}}{q^j_{(k-1)}} & \derpar{\tilde{q}^{\,i}_{(k)}}{p_j^{(0)}} &
\derpar{\tilde{q}^{\,i}_{(k)}}{p_j^{(1)}} & \ldots & \derpar{\tilde{q}^{\,i}_{(k)}}{p_j^{(k-2)}} &
\derpar{\tilde{q}^{\,i}_{(k)}}{p_j^{(k-1)}}
\end{array}
\right) \, ,
\end{equation*}
where every entry is a $n \times n$ real matrix, and every block of the matrix has size $kn \times kn$.
In particular, $\Id_n$ denotes the $n \times n$ identity matrix, $\mathbf{0}_n$ the $n \times n$
null matrix and, in the last row, we have $1 \leqslant i,j \leqslant n$. Observe that, in the most
favorable case, the map $\Leg H \colon T^*T^{(k-1)}Q \to TT^{(k-1)}Q$ has rank $(k+1)n = \dim T^{(k)}Q$,
that is, at the best the map $\Leg H_o \colon T^*T^{(k-1)}Q \to T^{(k)}Q$ is a submersion
onto $T^{(k)}Q$. A long but straightforward calculation shows that
\begin{equation*}
\left( \derpars{H}{p_i^{(k-1)}}{p_j^{(k-1)}} \right) =
\left( \derpar{\tilde{q}^{\,i}_{(k)}}{p_j^{(k-1)}} \right)
= \left( \derpar{\left( q^i_{(k)} \circ \leg_{L}^{-1} \right)}{p_j^{(k-1)}} \right) =
\left( \derpars{{L}}{q^i_{(k)}}{q^j_{(k)}} \right)^{-1} \, ,
\end{equation*}
which is a well-defined $n \times n$ matrix because ${L}$ is (hyper)regular, and therefore the map
$\Leg H_o$ is a submersion onto $T^{(k)}Q$. Observe that, as a consequence, $\Leg H_o$ admits local
sections, that is, maps $\sigma \colon T^{(k)}Q \to T^*T^{(k-1)}Q$ satisfying
$\Leg H_o \circ \sigma = \Id_{T^{(k)}Q}$.

The following result gives the illuminating key idea for the ``true'' definition of higher-order Hamiltonian function and regularity of a higher-order Hamiltonian function that we will exploit in the next Subsection.

\begin{theorem} Let $L:T^{(k)}Q\to\mathbb{R}$ be a $kth$-order Lagrangian function
\begin{itemize}
\item[(i)] If $L$ is regular, there exists an open subset $U\subset T^{(2k-1)}Q$ such that \begin{equation*}
\restric{\Leg H_o}{U} = \restric{\left( \tau_Q^{(k,2k-1)} \circ \leg_{L}^{-1} \right)}{U} \, .
\end{equation*}
\item[(ii)]If $L$ is hyperregular then $\Leg H_o:T^{*}T^{(k-1)}Q\to T^{(k)}Q$ is a submersion onto $T^{(k)}Q$.
\item[(iii)] If $L$ is hyperregular then $\Leg H_o$ admits a global section $\Upsilon:T^{(k)}Q\to T^{*}T^{(k-1)}Q$.
\end{itemize}
\end{theorem}
\begin{proof}
Let us consider the map $\Leg H_o \circ \leg_{L} \colon T^{(2k-1)}Q \to T^{(k)}Q$. In the
natural coordinates $\left( q^i_{(j)} \right)$ of $T^{(2k-1)}Q$ ($1 \leqslant i \leqslant n$,
$0 \leqslant j \leqslant k$), the local expression of this map is given by
\begin{align*}
\left( \Leg H_o \circ \leg_{L} \right) \left( q^i_{(0)},\ldots,q^i_{(2k-1)} \right) &=
\Leg H_o \left( q^i_{(0)},\ldots,q^i_{(k-1)};\hat{p}_i^{(0)},\ldots,\hat{p}_i^{(k-1)} \right) \\
&= \left( q^i_{(0)},\ldots,q^i_{(k-1)},q^i_{(k)} \right) \, ,
\end{align*}
since $\tilde{q}^{\,i}_{(k)} = (\leg_{L}^{-1})^*q_{(k)}^i$. From this coordinate expression we
deduce that for every point $[\gamma]^{(2k-1)}_0 \in T^{(2k-1)}Q$ there exists an open subset
$U \subseteq T^{(2k-1)}Q$ such that $[\gamma]^{(2k-1)}_0 \in U$ and
\begin{equation*}
\restric{\left( \Leg H_o \circ \leg_{L} \right)}{U} = \restric{\tau_Q^{(k,2k-1)}}{U} \, ,
\end{equation*}
and, since ${L} \in C^\infty(T^{(k)}Q)$ is a regular Lagrangian function, the Legendre-Ostrogradsky
map is a local diffeomorphism, we have
\begin{equation*}
\restric{\Leg H_o}{U} = \restric{\left( \tau_Q^{(k,2k-1)} \circ \leg_{L}^{-1} \right)}{U} \, .
\end{equation*}
Observe that, moreover, we assume that ${L}$ is a hyperregular Lagrangian function. Therefore, the
map $\leg_{L}^{-1} \colon T^*T^{(k-1)}Q \to T^{(2k-1)}Q$ is bijective and defined in the entire
manifold $T^*T^{(k-1)}Q$. On the other hand, the map $\tau_Q^{(k,2k-1)}$ is surjective, from where
we deduce that $\Leg H_o \colon T^*T^{(k-1)}Q \to T^{(k)}Q$ is surjective in this case, since we have
\begin{equation*}
\Leg H_o = \tau_Q^{(k,2k-1)} \circ \leg_{L}^{-1} \, ,
\end{equation*}
that is, $\Im(\Leg H_o) = T^{(k)}Q$.

Observe that, in addition, there exists a map $\Upsilon \colon T^{(k)}Q \to
T^*T^{(k-1)}Q$ defined by $\Upsilon = \leg_{L} \circ \Psi$, with
$\Psi \in \Gamma\left( \tau_Q^{(k,2k-1)} \right)$ being a global section of $\tau_Q^{(k,2k-1)}$,
which satisfies
\begin{equation*}
\Leg H_o \circ \Upsilon
= \tau_Q^{(k,2k-1)} \circ \leg_{L}^{-1} \circ \leg_{L} \circ \Psi
= \tau_Q^{(k,2k-1)}  \circ \Psi = \Id_{T^{(k)}Q} \, ,
\end{equation*}
that is, $\Upsilon$ is a global section of $\Leg H_o$.
\end{proof}
\subsection{Higher-order Hamiltonian functions: definition and regularity}

Bearing in mind the results obtained in the previous Section, we are now able to give a solution to
Problems \ref{prob:HOHamiltonianFunctions} and \ref{prob:HORegularHamiltonianFunctions}. From the
last theorem the statement of Problem \ref{prob:HOHamiltonianFunctions} and the calculations
for the particular case of a Hamiltonian system associated to a hyperregular Lagrangian system, we
can give the following definition.

\begin{definition}\label{def:HOHamiltonianFunction}
A function $H \in C^\infty(T^*T^{(k-1)}Q)$ is a \textnormal{$k$th-order Hamiltonian function} if
its fiber derivative, $\Leg H \colon T^*T^{(k-1)}Q \to TT^{(k-1)}Q$, takes values in the
``holonomic'' submanifold $j_k \colon T^{(k)}Q \hookrightarrow TT^{(k-1)}Q$, that is,
$\Im(\Leg H) \subseteq j_k(T^{(k)}Q)$.
\end{definition}

In the induced natural coordinates $\left( q^i_{(j)},p_i^{(j)} \right)$, $1 \leqslant i \leqslant n$
and $0 \leqslant j \leqslant k-1$, of $T^*T^{(k-1)}Q$, taking into account that the submanifold
$j_k \colon T^{(k)}Q \to TT^{(k-1)}Q$ is defined locally by the $(k-1)n$ constraints
$v_{(j)}^i = q_{(j+1)}^i$ ($0 \leqslant j \leqslant k-2$) and that the fiber
derivative of $H$ is a fiber bundle morphism, the condition for $H$ to be a $k$th-order Hamiltonian
function gives in coordinates
\begin{equation}\label{eqn:HOHamiltonianFunctionLocal}
\derpar{H}{p^{(j)}_i} = \Leg H^*v_{(j)}^i = \Leg H^*q_{(j+1)}^i = q_{(j+1)}^i \, , \
\mbox{for every } 1 \leqslant i \leqslant n \, , \ 0 \leqslant j \leqslant k-2 \, .
\end{equation}

Observe that if $H \in C^\infty(T^*T^{(k-1)}Q)$ is a $k$th-order Hamiltonian function, then the
fiber derivative $\Leg H$ of $H$ induces a map $\Leg H_o \colon T^*T^{(k-1)}Q \to T^{(k)}Q$ defined
as $\Leg H = j_k \circ \Leg H_o$. This map is given in coordinates by
\begin{equation}\label{eqn:InducedFiberDerivativeHam}
\Leg H_o \left( q^i_{(0)},\ldots,q^i_{(k-1)};p_i^{(0)},\ldots,p_i^{(k-1)} \right)
= \left( q^i_{(0)},\ldots,q^i_{(k-1)}, \derpar{H}{p_i^{(k-1)}} \right) \, .
\end{equation}
The map $\Leg H_o \colon T^*T^{(k-1)}Q \to T^{(k)}Q$ enables us to give a solution to Problem
\ref{prob:HORegularHamiltonianFunctions} in terms of its (local) ``inverse map''. Nevertheless,
due to the restriction imposed by the dimensions of the manifolds involved, there is no inverse map
to $\Leg H_o$ (even locally). Hence, instead of an inverse, we use the most similar approach, which
consists in considering inverses in just one way, that is, sections. Therefore, following the
patterns in \cite{art:Saunders_Crampin90}, we give the following definition.

\begin{definition}\label{def:RegularHOHamiltonian}
A $k$th-order Hamiltonian function $H \in C^\infty(T^*T^{(k-1)}Q)$ is said to be \textnormal{regular}
if the map $\Leg H_o \colon T^*T^{(k-1)}Q \to T^{(k)}Q$ is a submersion onto $T^{(k)}Q$. If, moreover,
$\Leg H_o$ admits a global section $\Upsilon \colon T^{(k)}Q \to T^*T^{(k-1)}Q$, the Hamiltonian
function is said to be \textnormal{hyperregular}. Otherwise, the Hamiltonian function is said to be
\textnormal{singular}.
\end{definition}

In the induced natural coordinates $\left( q^i_{(j)},p_i^{(j)} \right)$, $1 \leqslant i \leqslant n$
and $0 \leqslant j \leqslant k-1$, of $T^*T^{(k-1)}Q$ from the relation
\eqref{eqn:HOHamiltonianFunctionLocal} we have
\begin{equation}\label{eqn:CrossedMomentaDerivativeHam}
\derpars{H}{p^{(s)}_j}{p_i^{(k-1)}} = \derpar{}{p_i^{(k-1)}} \left( \derpar{H}{p^{(s)}_j} \right)
= \derpar{q_{(s+1)}^j}{p_i^{(k-1)}} = 0 \, ,
\end{equation}
where $1 \leqslant i,j \leqslant n$ and $0 \leqslant s \leqslant k-2$. From those and the local
expression \eqref{eqn:InducedFiberDerivativeHam} of $\Leg H_o$ we deduce that its tangent
map in an arbitrary point $\alpha_q \in T^*T^{(k-1)}Q$ is given locally by the following
$(k+1)n \times 2kn$ real matrix
\begin{equation}\label{eqn:TangentInducedFiberDerivativeHam}
T_{\alpha_q}\Leg H_o =
\left(
\begin{array}{ccc|cccc}
\Id_n & \ldots & \mathbf{0}_n & \mathbf{0}_n & \ldots & \mathbf{0}_n & \mathbf{0}_n \\
\vdots & \ddots & \vdots & \vdots & \ddots & \vdots & \vdots \\
\mathbf{0}_n & \ldots & \Id_n & \mathbf{0}_n & \ldots & \mathbf{0}_n & \mathbf{0}_n \\ \hline
\ & \ & \ & \ & \ & \ \\[-10pt]
\derpars{H}{q^j_{(0)}}{p_i^{(k-1)}} & \ldots & \derpars{H}{q^j_{(k-1)}}{p_i^{(k-1)}} &
\mathbf{0}_n & \ldots & \mathbf{0}_n & \derpars{H}{p^{(k-1)}_j}{p_i^{(k-1)}}
\end{array}
\right) \, ,
\end{equation}
where every entry is a $n \times n$ real matrix. In particular, $\Id_n$ denotes the $n \times n$
identity matrix, $\mathbf{0}_n$ the $n \times n$ null matrix and, in the last row, we have
$1 \leqslant i,j \leqslant n$. Therefore, the local condition for a $k$th-order Hamiltonian function
$H \in C^\infty(T^*T^{(k-1)}Q)$ to be a regular is
\begin{equation*}
\det \left( \derpars{H}{p^{(k-1)}_j}{p_i^{(k-1)}} \right) (\alpha_q) \neq 0 \, , \
\mbox{for every } \alpha_q \in T^*T^{(k-1)}Q \, .
\end{equation*}

Using both Definitions \ref{def:HOHamiltonianFunction} and \ref{def:RegularHOHamiltonian} we can
now state and prove the analogous results to Propositions \ref{prop:LagrangianFromHamiltonian} and
\ref{prop:FiberDerivativeHInverse} in the higher-order setting. As for Proposition
\ref{prop:LagrangianFromHamiltonian}, we have the following result.

\begin{proposition}\label{prop:HOLagrangianFromHamiltonian}
Let $H \in C^\infty(T^*T^{(k-1)}Q)$ be a regular $k$th-order Hamiltonian function,
$\theta_{T^{(k-1)}Q} \in \df^{1}(T^*T^{(k-1)}Q)$ the Liouville 1-form,
$\omega_{T^{(k-1)}Q} \in \df^{2}(T^*T^{(k-1)}Q)$ the canonical symplectic form, and
$X_H \in \vf(T^*T^{(k-1)}Q)$ the unique vector field solution to the equation
\begin{equation}\label{eqn:DynamicalEquationHOHam}
\inn_{X_H}\omega_{T^{(k-1)}Q} = \d H \, .
\end{equation}
Then the function $\theta_{T^{(k-1)}Q}(X_H) - H \in C^\infty(T^*T^{(k-1)}Q)$ is $\Leg H_o$-projectable,
and the function ${L} \in C^\infty(T^{(k)}Q)$ defined by
$(\Leg H_o)^*{L} = \theta_{T^{(k-1)}Q}(X_H) - H$ is a regular $k$th-order Lagrangian function.
Moreover, for every $[\gamma]^{(2k-1)}_0 \in T^{(2k-1)}Q$ there exists an open set
$U \subseteq T^{(2k-1)}Q$ such that $[\gamma]^{(2k-1)}_0 \in U$ and
$\restric{\left( \Leg H_o \circ \leg_{L} \right)}{U} = \restric{\tau_Q^{(k,2k-1)}}{U}$.
\end{proposition}
\begin{proof}
Following the patterns in \cite{book:Abraham_Marsden78}, the easiest proof of this result is done
in coordinates. Hence, let $\left( q^i_{(j)},p_i^{(j)} \right)$, $1 \leqslant i \leqslant n$
and $0 \leqslant j \leqslant k-1$, be the induced natural coordinates in a suitable open set of
$T^*T^{(k-1)}Q$. Along this proof we only consider these coordinates.

First, we must prove that the function $\theta_{T^{(k-1)}Q}(X_H) - H \in C^\infty(T^*T^{(k-1)}Q)$
is $\Leg H_o$-projectable, that is,
\begin{equation*}
\Lie(Y)\left( \theta_{T^{(k-1)}Q}(X_H) - H \right) = 0 \, , \
\mbox{for every } Y \in \ker T\Leg H_o \, .
\end{equation*}
From the coordinate expression \eqref{eqn:TangentInducedFiberDerivativeHam} of the tangent map
$T\Leg H_o$ at a point $\alpha_q \in T^*T^{(k-1)}Q$, it is clear that a local basis for
$\ker T\Leg H_o$ is given by
\begin{equation*}
\ker T\Leg H_o = \left\langle \derpar{}{p_i^{(0)}},\ldots,\derpar{}{p_i^{(k-2)}} \right\rangle \, .
\end{equation*}
On the other hand, from our calculations in Section \ref{subsec:StatementProblems} we know that
the Hamiltonian vector fiels solution to equation \eqref{eqn:DynamicalEquationHOHam} is given
locally by \eqref{eqn:VectorFieldSolutionHamEq}, which in combination with the identities
\eqref{eqn:HOHamiltonianFunctionLocal} gives
\begin{equation*}
X_H = \sum_{l=0}^{k-2} q_{(l+1)}^i\derpar{}{q_{(l)}^i}
+ \derpar{H}{p_i^{(k-1)}} \derpar{}{q^i_{(k-1)}} - \derpar{H}{q^i_{(j)}} \derpar{}{p_i^{(j)}} \, .
\end{equation*}
Thus, since the Liouville $1$-form is given in coordinates by
$\theta_{T^{(k-1)}Q} = p_i^{(j)} \d q_{(j)}^i$, the function
$\theta_{T^{(k-1)}Q}(X_H) - H \in C^\infty(T^*T^{(k-1)}Q)$ has the following coordinate expression
\begin{equation*}
\theta_{T^{(k-1)}Q}(X_H) - H =
\sum_{l=0}^{k-2} p_i^{(l)} q_{(l+1)}^i + p_i^{(k-1)}\derpar{H}{p_i^{(k-1)}} - H \, .
\end{equation*}
Hence for every $Y = \partial / \partial p_r^{(s)}$, $1 \leqslant r \leqslant n$ and
$0 \leqslant s \leqslant k-2$, we have
\begin{align*}
\Lie(Y)\left( \theta_{T^{(k-1)}Q}(X_H) - H \right) &=
\Lie\left( \derpar{}{p_r^{(s)}} \right)\left( \sum_{l=0}^{k-2} p_i^{(l)} q_{(l+1)}^i
+ p_i^{(k-1)}\derpar{H}{p_i^{(k-1)}} - H \right) \\
&= q_{(s+1)}^r + p_i^{(k-1)} \derpars{H}{p_r^{(s)}}{p_i^{(k-1)}} - \derpar{H}{p_r^{(s)}} = 0 \, ,
\end{align*}
where, in the last equality, we have used the identities \eqref{eqn:HOHamiltonianFunctionLocal}
and \eqref{eqn:CrossedMomentaDerivativeHam}. Therefore, the function 
$\theta_{T^{(k-1)}Q}(X_H) - H \in C^\infty(T^*T^{(k-1)}Q)$ is $\Leg H_o$-projectable.

It is now a long but straightforward calculation to prove that the Lagrangian function
$L \in C^\infty(T^{(k)}Q)$ defined by $(\Leg H_o)^*{L} = \theta_{T^{(k-1)}Q}(X_H) - H$ is regular.
Indeed, using implicit differentiation and the chain rule, we have
\begin{equation*}
\left( \derpars{L}{q_{(k)}^i}{q_{(k)}^j} \right) = \left( \derpar{\hat{p}^{(k-1)}_j}{q^i_{(k)}} \right)
= \left( \derpar{p^j_{(k-1)} \circ \leg_L}{q^i_{(k)}} \right) =
\left( \derpars{H}{p_j^{(k-1)}}{p_i^{(k-1)}} \right)^{-1} \, ,
\end{equation*}
where $\hat{p}^{(k-1)}_j$ are $n$ last Jacobi-Ostrogradsky momenta functions defined in
\eqref{eqn:JacobiOstrogradskyMomenta}. Therefore, since $H$ is a $k$th-order regular Hamiltonian
function, the Hessian of $H$ with respect to the highest order momenta is invertible at every point
of $T^*T^{(k-1)}Q$, and hence so is the Hessian of $L$ with respect to the highest order velocities.
Thus, $L$ is regular.

Finally, let us consider the map $\Leg H_o \circ \leg_{L} \colon T^{(2k-1)}Q \to T^{(k)}Q$. In the
natural coordinates of $T^{(2k-1)}Q$ the local expression of this map is given by
$$
\left( \Leg H_o \circ \leg_{L} \right) \left( q^i_{(0)},\ldots,q^i_{(2k-1)} \right) =
\Leg H_o \left( q^i_{(0)},\ldots,q^i_{(k-1)},\hat{p}_i^{(0)},\ldots,\hat{p}_i^{(k-1)} \right) 
= \left( q^i_{(0)},\ldots,q^i_{(k-1)},q^i_{(k)} \right)$$
since $\tilde{q}^{\,i}_{(k)} = \Leg H_o^*q_{(k)}^i$. From this coordinate expression we
deduce that for every point $[\gamma]^{(2k-1)}_0 \in T^{(2k-1)}Q$ there exists an open set
$U \subseteq T^{(2k-1)}Q$ such that $[\gamma]^{(2k-1)}_0 \in U$ and
$\restric{\left( \Leg H_o \circ \leg_{L} \right)}{U} = \restric{\tau_Q^{(k,2k-1)}}{U}$.
\end{proof}

\begin{remark}
Since $\Leg H_o$ is at the best a submersion we can not directly pull-back a function in
$T^*T^{(k-1)}Q$ to $T^{(k)}Q$ by the inverse of $\Leg H_o$ (since there isn't), and hence we must
first prove the $\Leg H_o$-projectability of the function.
\end{remark}

\begin{remark}
It is important to point out that in the conclusion of Proposition \ref{prop:HOLagrangianFromHamiltonian}
we only prove the \textit{regularity} of $L$, instead of its hyperregularity (as in Proposition
\ref{prop:LagrangianFromHamiltonian}). Indeed, even if the Hamiltonian function is hyperregular, we
can only prove that $\leg_L$ is a surjective local diffeomorphism, but not the injectivity.
\end{remark}

Before stating and proving the analogous to Proposition \ref{prop:FiberDerivativeHInverse}, we need
the following technical result, which we state and prove in a more general case than we need.

\begin{lemma}\label{lemma:TechLemma}
Let $L \in C^\infty(T^{(k)}Q)$ be a $k$th-order Lagrangian function,
$\theta_L \in \Omega^1(T^{(2k-1)}Q)$ the associated Poincar\'e-Cartan $1$-form and $E_L \in
C^\infty(T^{(2k-1)}Q)$ the $k$th-order Lagrangian energy. If $X \in \vf(T^{(2k-1)}Q)$ is a semispray
of type $k$, then $\left(\tau_Q^{(k,2k-1)}\right)^*L = \theta_L(X) - E_L$.
\end{lemma}
\begin{proof}
This proof is easy in coordinates. Recall that the coordinate expression of a semispray of type $r$
in $T^{(k)}Q$ is given by \eqref{eqn:SemisprayTypeRLocal}, from where we deduce that a semispray
of type $k$ $X$ in $T^{(2k-1)}Q$ is given locally by
\begin{equation*}
X = q_{(1)}^i\derpar{}{q_{(0)}^i} + \ldots + q_{(k)}^i\derpar{}{q_{(k-1)}^i} +
F_{(k)}^i\derpar{}{q_{(k)}^i} + \ldots + F_{(2k-1)}^i\derpar{}{q_{(2k-1)}^i} \, .
\end{equation*}
Hence, bearing in mind the coordinate expression \eqref{eqn:PoincareCartan1FormLocal}
of the Poincar\'{e}-Cartan $1$-form, we deduce that the smooth function $\theta_L(X)$ is given
locally by
\begin{equation*}
\theta_L(X) = \sum_{r=1}^k q_{(r)}^i \sum_{j=0}^{k-r}(-1)^j d_T^j\left(\derpar{L}{q_{(r+j)}^i}\right) \, .
\end{equation*}
Hence, bearing in minda the coordinate expression \eqref{eqn:LagrangianEnergyLocal} of the Lagrangian
energy $E_L$, we obtain
\begin{align*}
\theta_L(X) - E_L &= \sum_{r=1}^k q_{(r)}^i \sum_{j=0}^{k-r}(-1)^j d_T^j\left(\derpar{L}{q_{(r+j)}^i}\right)
- \sum_{r=1}^k q_{(r)}^i \sum_{j=0}^{k-r}(-1)^j d_T^j\left(\derpar{L}{q_{(r+j)}^i}\right)
+ \left( \tau_Q^{(k,2k-1)} \right)^*L \\
&= \left( \tau_Q^{(k,2k-1)} \right)^*L \, .
\end{align*}
as claimed.
\end{proof}

Finally, as for Proposition \ref{prop:FiberDerivativeHInverse}, we have the following result.

\begin{proposition}\label{prop:HOFiberDerivativeHInverse}
Let ${L} \in C^\infty(T^{(k)}Q)$ be a hyperregular $k$th-order Lagrangian function and
$H = E_{L} \circ \leg_{L}^{-1} \in C^\infty(T^*T^{(k-1)}Q)$ the associated Hamiltonian function.
Then $H$ is a hyperregular $k$th-order Hamiltonian function and
$\Leg H_o = \tau_Q^{(k,2k-1)} \circ \leg_{L}^{-1}$. In addition, if
$\tilde{L} \in C^\infty(T^{(k)}Q)$ is the $k$th-order Lagrangian function associated
to $H$ by Proposition \ref{prop:HOLagrangianFromHamiltonian}, then $\tilde{{L}} = {L}$.
\end{proposition}
\begin{proof}
The proof of $H$ being a hyperregular $k$th-order Hamiltonian functions follows exactly the same
patterns as the computations carried out in Section \ref{subsec:ParticularCase}, and hence we omit
them. Therefore, the only need to prove that if $\tilde{L} \in C^\infty(T^{(k)}Q)$ is the $k$th-order
Lagrangian function associated to $H$ by Proposition \ref{prop:HOLagrangianFromHamiltonian}, then
$\tilde{{L}} = {L}$. Observe that from the properties of the map $\Leg H_o^*$ we have
\begin{equation*}
\Leg H_o^*\tilde{L} = (\tau_Q^{(k,2k-1)} \circ \leg_L^{-1})^*\tilde{L}
= \left(\leg_L^{-1}\right)^*\left( \left(\tau_Q^{(k,2k-1)}\right)^*\tilde{L} \right) \, .
\end{equation*}
On the other hand, from the definition of the Hamiltonian function and the properties of the
Legendre-Ostrogradsky map we have
\begin{align*}
\theta_{T^{(k-1)}Q}(X_H) - H &= \theta_{T^{(k-1)}Q}(X_H) - (E_L \circ \leg_L^{-1})
= \left( \leg_L^{-1} \right)^*\left( \theta_{T^{(k-1)}Q}(X_H) \circ \leg_L - E_L \right) \\
&= \left( \leg_L^{-1} \right)^*\left( \theta_L(X_L) - E_L \right)
= \left( \leg_L^{-1} \right)^*\left( \left( \tau_Q^{(k,2k-1)}\right)^*L \right) \, ,
\end{align*}
where in the last step we have used Lemma \ref{lemma:TechLemma}, as the vector field $X_L$
solution to the Lagrangian dynamical equation $\inn_{X_L}\omega_L = \d E_L$ is a semispray of type
$1$ when $L$ is a hyperregular higher-order Lagrangian function (see \cite{book:DeLeon_Rodrigues85}
for details). Equating these two expressions, we have
\begin{equation*}
(\leg_L^{-1})^*\left( \left(\tau_Q^{(k,2k-1)}\right)^*\tilde{L} \right)
= \left( \leg_L^{-1} \right)^*\left( \left( \tau_Q^{(k,2k-1)}\right)^*L \right) \, .
\end{equation*}
Now, since $\leg_L^{-1}$ is a global diffeomorphism and $\tau_Q^{(k,2k-1)}$ is the canonical
projection, this last equality holds if, and only if, $\tilde{L} = L$.
\end{proof}

\section*{Acknowledgments} 

This work has been partially supported by NSF grant INSPIRE-1363720, \textsl{Ministerio de
Econom\'{\i}a y Competividad} (Spain) grants MTM2013-42870-P and MTM2014-54855-P, and 
\textsl{Generalitat de Catalunya} (Catalonia) project 2014-SGR-634.
We would like to thank D. Mart\'{\i}n de Diego and N. Rom\'an Roy for providing the initial stimulus
for this project. Also we wish to thanks M. de Le\'on, J.C. Marrero and M.C. Mu\~{n}oz Lecanda for
fruitful comments and discussions. 


\end{document}